\documentclass[preprint,12pt]{elsarticle}

\biboptions{numbers,sort&compress}

\usepackage{amssymb}
\usepackage{amsmath}
\usepackage{amsthm}

\usepackage{bbold}
\usepackage{booktabs}
\usepackage{comment}
\usepackage{enumitem}	
\usepackage{hyperref}
\usepackage{natbib}
\usepackage{subfig}
\usepackage{tabularx}
\usepackage{tikz}
\usepackage{xcolor}

\usetikzlibrary{positioning,arrows.meta}

\journal{Communications in Nonlinear Science and Numerical Simulation}

\theoremstyle{plain}%
\newcommand\beq{\begin{equation}}%
\newcommand\eeq{\end{equation}}%
\newtheorem{defi}{Definition}%
\newtheorem{theo}{Theorem}%
\newtheorem{cor}{Corollary}%
\theoremstyle{remark}%
\newtheorem{rem}{Remark}%
\newcommand\bb{\mathbb}%
\DeclareMathOperator{\diag}{diag}%
\DeclareMathOperator{\rank}{rank}%

\begin{document}

\begin{frontmatter}



\title{On a cross-coupling of Rulkov neural maps}

\author[1]{Stefano Disca\corref{cor1}}
\ead{dscsfn@unife.it}
\cortext[cor1]{Corresponding author}

\affiliation[1]{organization={Department of Mathematics and Computer Science, University of Ferrara},
addressline={Via Machiavelli 30},
city={Ferrara},
postcode={44121},
country={Italy}}

\begin{abstract}
We introduce a novel coupling of Rulkov neural maps, proposing a heuristic biological interpretation for the transition to non-small values of the perturbations acting on the slow variables. We analytically prove that the coupling of two identical neurons preserves the emergence of Devaney chaos through the existence of a generalized snap-back repeller, provided that a snap back repeller exists for the original system. We present numerical simulations for the coupling of two different neurons showing the arising of a potential global strange attractor, whose fractal structure is strongly suggested by the computation of a non-integer Kaplan-Yorke dimension. Furthermore, we perform standard numerical studies concerning time series, Lyapunov exponents spectra, bifurcation diagrams and basins of attraction. Finally, we briefly propose a generalization of the coupling to an arbitrary number of neurons.
\end{abstract}



\begin{keyword}
coupled Rulkov neural maps \sep snap-back repeller \sep Marotto theorem \sep Shiraiwa-Kurata theorem \sep Kaplan-Yorke dimension

\MSC[2020]{37D10; 37D45; 37N25; 39A28; 39A33}
\end{keyword}

\end{frontmatter}



\section{Introduction}\label{sec_introduction}
The modeling of the nervous system must rely on the description of its basic element, the neuron. Since the formulation of the celebrated Hodgkin-Huxley model \cite{Hodgkin1952}, incredible efforts have been done in order to formulate precise mathematical models of the brain. In the last decades, together with models based on ODEs and PDEs, interest has been gone through simple discrete maps, that could provide efficient predictions under very cheap computational performances. 

Beyond other important models \cite{Chialvo1995, Izhikevich2003}, the Rulkov map \cite{Rulkov2002} represents a remarkable example of a simple two-dimensional discrete system able to qualitatively reproduce spiking and bursting of a real neuron. The original model defined by Rulkov can be taken in the form
\beq\begin{cases}\label{mapR}
x_{n+1} = \alpha f(x_n) + y_n \\
y_{n+1} = y_n - \mu(x_n - \sigma) \,,
\end{cases}\eeq
where $\alpha > 0$, $0 < \mu \ll 1$, $\sigma \in \bb{R}$. In \eqref{mapR}, the \textit{fast variable} $x_n$ is the dynamical variable of the system, that evolves according to the iterations of the \textit{slow variable} $y_n$. The function $f(x_n)$, together with the parameter $\alpha$, models the membrane potential acting on the neuron. The parameter $\sigma$ can be interpreted as an external current acting on the neuron, while the assumption $\mu \ll 1$ is responsible for the slow evolution of $y_n$, being the main source of spiking and bursting in the time series of $x_n$. Among the several possibilities, the particular choice $f(x_n) = \frac{1}{1 + x_n^2}$ in \eqref{mapR} gives rise to a fascinating and rich behavior, including transition to chaos \cite{Ibarz2011}.

In this work we introduce a novel coupling of Rulkov maps and analytically prove that the coupling preserves the emergence of Devaney chaos, if it arises in the original two-dimensional system. This coupling is structurally different from the usual ones analyzed in literature and it appears preserving the main dynamical phenomena exhibited by the Rulkov map in the classical regime on the parameters, while some interesting properties arise in the other cases.

The paper is organized as follows. In Section \ref{sec_model} we introduce the coupling subject of the work, proposing a heuristic biological interpretation for the transition to non-small values of the perturbations acting on the slow variables. In Section \ref{sec_theo} we introduce the mathematical framework and state the main theorem, proving that the coupling of two identical Rulkov neurons admits a generalized snap-back repeller that is explicitly constructed from the snap-back repeller associated to the Rulkov map. In Section \ref{sec_num} we perform some numerical studies on the coupling of two chaotic Rulkov maps. In Section \ref{sec_num_attractor} we show the arising of a potential strange attractor, whose fractal structure is strongly suggested by the computation of the Kaplan-Yorke dimension. In Section \ref{sec_num_lyap} we perform standard numerical studies, presenting spectra of Lyapunov exponents, bifurcation diagrams and basins of attraction. In Section \ref{sec_conclusions} we state our conclusions and propose further developments of this work, also suggesting a natural generalization of the coupling to a large number of neurons. In \ref{sec_app} we discuss some computational details omitted in the proof of the main theorem.

\section{Cross coupling of Rulkov maps}\label{sec_model}
Since its formulation, there have been several efforts to extend the Rulkov map to more interacting neurons, leading to an extensive literature on coupled Rulkov maps (for a pair of examples, see \cite{Rakshit2018} and \cite{Mirzaei2022}). However, to the best of our knowledge, all the couplings analyzed in literature always start from two or more separate 2D models that are coupled through a diffusive term taking into account the interaction between the neurons. Here, we propose a new formal way to couple two Rulkov maps that is structurally different from the standard couplings. The novelty of the model lies on the coupling itself, the heuristic biological interpretation we propose and some interesting features encountered in the non-standard case of non-small perturbations acting on the slow variable.

Given two discrete maps defined by \eqref{mapR}, we define their \textit{cross-coupling} as
\beq\begin{cases}\label{mapC}
x_{n+1} = \alpha f(z_n) + y_n \\
y_{n+1} = y_n - \mu(x_n - \sigma) \\
z_{n+1} = \beta f(x_n) + \omega_n \\
\omega_{n+1} = \omega_n - \nu(z_n - \rho) \,,
\end{cases}\eeq
where $\alpha, \beta > 0$, $\mu, \nu \in (0, 1)$, $\sigma, \rho \in \bb{R}$.

In view of the numerical results presented in the work for non-small values of $\mu$, $\nu$, we propose the following heuristic biological interpretation of \eqref{mapC}.
\begin{itemize}
\item For small values of $\mu, \nu$, i.e. $\mu, \nu \ll 1$, \eqref{mapC} describes a cross-coupling of two Rulkov maps, in the sense that the evolution of the fast variable of one neuron at a certain step is determined by the evolution of the other one at the previous step, while the slow variable evolves according to the associate fast variable.
\item For non-small values of $\mu, \nu$, i.e. $\mu, \nu \lesssim 1$, the variables $y$, $\omega$ stop acting as slow variables and \eqref{mapC} is meant to model the interaction between four neurons belonging to two distinct populations, i.e. two distinct brain regions that behave via a linear and nonlinear interaction, respectively. In particular, the four neurons are meant belonging to the frontier of these brain regions.
\end{itemize}
A schematic representation of the previous interpretation is shown in Figure \ref{switch}. It is worth to notice that the model \eqref{mapC} does not allow decoupling of the neurons for any values on the parameters.

\begin{figure}[h!]
\centering
\begin{tikzpicture}[>=latex]
\node (A) at (0, 0)
{
$(\underbrace{x}_{ \text{fast} } \,, \underbrace{y}_{ \text{slow} })$ = \text{neuron XY}
};
\node[below=0.005cm of A]
{
$(\underbrace{z}_{ \text{fast} } \,, \underbrace{\omega}_{ \text{slow} })$ = \text{neuron ZW}
};
\node[below=2cm of A] (B)
{
\tikz{
	\begin{scope}[name prefix = top-]
	\node[draw] (A) at (0,0) {XY};
  	\end{scope}
  	\begin{scope}[name prefix = bottom-]
    	\node[draw] (A) at (2,0) {ZW};
  	\end{scope}
  	\draw[thick, <->] (top-A) -- (bottom-A);
}
};
\node[right=4cm of A] (C)
{
$(x \,, y \,, z \,, \omega) =$
};
\node[below=0.005cm of C, align=right]
{
\text{neuron $X$ + $Y$ + $Z$ + $W$}
};
\node[below=1.25cm of C] (D)
{
\tikz{
	\begin{scope}[name prefix = top-]
	\node[draw] (A) at (0,2) {X};
    	\node[draw] (B) at (2,2) {Y};
    	\draw[thick, <->] (A) -- (B);
  	\end{scope}
  	\begin{scope}[name prefix = bottom-]
    	\node[draw] (A) at (0,0) {Z};
    	\node[draw] (B) at (2,0) {W};
    	\draw[thick, <->] (A) -- (B);
  	\end{scope}
  	\draw[thick, <->] (top-A) -- (bottom-A);
}
};
\node (L) at (0, -5.25) {$\mu \,, \nu \ll 1$};
\node[right=6cm of L] (R) {$\mu \,, \nu \lesssim 1$};
\draw[very thick,->] (L.east) -- (R.west);
\end{tikzpicture}
\caption{A heuristic biological interpretation of \eqref{mapC}. For $\mu, \nu \ll 1$, $(x, y)$ and $(z, \omega)$ are two fast-slow systems modeling two neurons $XY$ and $ZW$, respectively. For $\mu, \nu \lesssim 1$, $x, y, z, \omega$ are four state variables modeling four neurons $X, Y, Z, W$, where the pairs $X, Z$ and $Y, W$ belong to two different brain regions.}
\label{switch}
\end{figure}
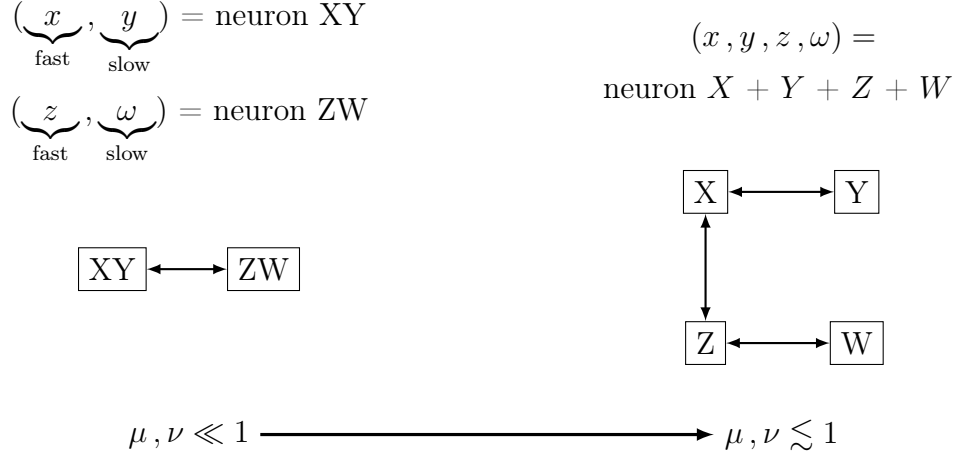

In the rest of this work, we focus on the cross-coupling of two chaotic Rulkov maps, i.e. we take in \eqref{mapR} and \eqref{mapC} $f(x) = \frac{1}{1 + x^2}$.

\section{Preservation of Devaney chaos under cross-coupling}\label{sec_theo}

\subsection{Boundedness of motion}
One standard way to prove the boundedness of the orbits for a discrete dynamical system relies on the computation of a Lyapunov function, that is a continuous function that is positive, coercive and decreasing along the orbits of the system. However, there is not an algorithmic way to construct a Lyapunov function and this procedure usually relies on the specific structure of the system, often representing a difficult task. Instead of doing that, we focus on the existence of an absorbing set for the system \eqref{mapC}, that implies boundedness of the orbits.

\begin{defi}[absorbing set]
Given a discrete map $f \colon X \to X$, $x_{n+1} = f(x_n)$, a compact set $P \subset X$ is an \emph{absorbing set} if $\forall x_0 \in X$ $\exists N_0 \ge 0$ such that $x_n \in P$ $\forall n \ge N_0$.
\end{defi}

\begin{theo}[Ortega, 2026]\label{theo_Ortega}
Consider the system
\beq\label{system_Ortega}
x_{n+1} = A x_n + b(x_n) \,, \quad x_n \in \bb{R}^n \,,
\eeq
where $A \in \bb{R}^{d \times d}$ is diagonalized with eigenvalues $\lambda_i$, $i = 1, \dots, d$. Suppose that
\beq
|\lambda_i| < 1 \quad \forall i = 1 \,, \dots \,, d
\eeq
and $\exists C > 0$ such that
\beq
\| b(x) \| \le C \quad \forall x \in \bb{R}^d \,.
\eeq
Then, there exists an absorbing set for \eqref{system_Ortega}.
\end{theo}
\begin{proof}
\textit{The author is grateful to prof. Rafael Ortega Ríos (University of Granada) for kindly communicating the following proof.} \\From the assumptions on $A$, $\exists P \in \bb{R}^{d \times d}$ such that
\beq
A = P D P^{-1} \,, \quad D = \diag \{ \lambda_1 \,, \dots \,, \lambda_d \} \,, \quad |\lambda_i| \le \mu < 1 \,.
\eeq
We define the norm $\| \cdot \|_*$ in $\bb{R}^d$ as
\beq
\| x \|_* = \| P^{-1}x \|_\infty \,, \quad \| \xi \|_\infty = \max \{ |\xi_i| \colon i = 1 \,, \dots \,, d \} \,.
\eeq
Then,
\beq
\| A x \|_* = \| P^{-1} A x \|_\infty = \| D P^{-1} x \|_\infty \le \mu \| P^{-1} x \|_\infty = \mu \| x \|_* \,.
\eeq
Furthermore, we have
\beq
\| x_{n+1} \|_* \le \| A x_n \|_* + \| b(x_n ) \|_* \le \mu \| x_n \|_* + C \| P^{-1} \| \,.
\eeq
Then, $y_n = \| x_n \|_*$ satisfies
\beq
y_{n+1} \le \mu y_n + C \| P^{-1} \| \,,
\eeq
implying that $\lim \sup_{n \to +\infty} y_n \le \frac{1}{1 - \mu} C \| P^{-1} \|$. Finally, the set
\beq
P = \bigg\{ x \in \bb{R}^d \colon \| x \|_* \le \frac{1}{1-\mu} C \| P^{-1} \| \bigg\}
\eeq
is an absorbing set for \eqref{system_Ortega}.
\end{proof}

\begin{cor}\label{cor_Ortega}
If $f \in C^0_b(\bb{R})$ and $\mu, \nu \in (0, \frac{1}{2})$, then \eqref{mapC} has an absorbing set.
\end{cor}
\begin{proof}
The system \eqref{mapC} can be written as $\overline{x}_{n+1} = A \overline{x}_n + b(\overline{x}_n)$, $\overline{x}_n \in \bb{R}^4$, where
\beq
A =
\begin{pmatrix}
0 & 1 & 0 & 0 \\
-\mu & 1 & 0 & 0 \\
0 & 0 & 0 & 1 \\
0 & 0 & -\nu & 1
\end{pmatrix} \,, \quad b(\overline{x}_n) =
\begin{pmatrix}
\alpha f(z_n) \\ \mu \sigma \\ \beta f(x_n) \\ \nu \rho
\end{pmatrix} \,.
\eeq
Eigenvalues of $A$ are given by
\beq
\lambda_{1,2} = \frac{1 \pm \sqrt{1 - 4 \mu} }{2} \,, \quad \lambda_{3,4} = \frac{1 \pm \sqrt{1 - 4 \nu} }{2} \,.
\eeq
If $\mu, \nu \in (0, \frac{1}{2})$, then all eigenvalues are lower than $1$ in magnitude. Furthermore, $\exists C > 0$ such that $\| b \| \le C$ $\forall \overline{x} \in \bb{R}^4$, since $f \in C^0_b(\bb{R})$. Therefore, by Theorem \ref{theo_Ortega} the system \eqref{mapC} has an absorbing set.
\end{proof}

Thanks to Theorem \ref{theo_Ortega}, Corollary \ref{cor_Ortega} provides sufficient conditions on the parameters such that \eqref{mapC} has an absorbing set and subsequently bounded orbits, regardless the analytical formulation of the map.

\subsection{The main theorem}
A natural analytical approach to study chaotic orbits of a discrete map is given by the Marotto theorem \cite{Marotto1978, Marotto2005}.

\begin{defi}[Marotto, 2005]\label{Marotto_defi}
Consider a discrete map $f \in C^1(\bb{R}^d)$, $f \colon \bb{R}^d \to \bb{R}^d$, $x_{k+1} = f(x_k)$, $k \in \bb{N}$. Let $D f(x_k)$ be the Jacobian of $f$ with eigenvalues $\lambda_i(x_k)$. Let $B_r(y) =$ $\{ x \in \bb{R}^d \colon |x - y| < r$, $|\lambda_i(x)| > 1$ $\forall i \}$ be a repelling neighborhood of a point $y$. A point $p_0$ is a \emph{snap-back repeller (SPR)} if:
\begin{itemize}
\item $f(p_0) = p_0$;
\item  $|\lambda_i(p_0)| > 1$ $\forall i = 1, \dots, n$;
\item $\exists q_0 \in B_r(p_0)$, $q_0 \ne p_0$ and $\exists k \ge 1$ such that $f^{(k)}(q_0) = p_0$;
\item $\det( Df^{(j)}(q_0) ) \ne 0$ $\forall j = 1, \dots, k$.
\end{itemize}
\end{defi}

\begin{theo}[Marotto, 1978; Marotto, 2005]\label{Marotto_theo}
If a discrete map $f$ has a SPR, then it is chaotic in the sense of Li-Yorke \cite{LiYorke1975}.
\end{theo}

\begin{rem}\label{rem_SPR}
It has been proved in \cite{ShiChen2005} that, under a different but equivalent definition, the existence of a SPR in a general Banach space implies chaos in the sense of Devaney \cite{Devaney2021, Banks1992}, when the map is restricted to a Cantor set.
\end{rem}

It is important to mention that the two-dimensional Rulkov map has a SPR for this choice on $f$, as shown in \cite{GeCao2021}. In particular, it is proved that \eqref{mapR} has SPR of order two when $\alpha = 15$, $\mu = 0.001$, $\sigma \in [-2, 0.1]$, where the definition of $\sigma$ is taken such that the repelling condition is satisfied. On the other hand, the Marotto theorem cannot be applied to \eqref{mapC}; indeed, we will see that in the regime of the parameters considered in this work, i.e. $\alpha, \beta > 0$, $\mu, \nu \in (0, 1)$, $\sigma, \rho \in \bb{R}$, the fixed point of \eqref{mapC} is always a saddle point. Therefore, we adopt the alternative formulation introduced by Shiraiwa and Kurata \cite{Shiraiwa1981}, that extends the concept of an SPR to a generic hyperbolic fixed point by imposing precise conditions on the structure of the stable and unstable manifolds associated with the fixed point of the map.

Here, we recall the definition and the main theorem presented in \cite{Shiraiwa1981}.
\begin{defi}[Shiraiwa-Kurata, 1981]\label{Shiraiwa_defi}
Let $M$ be a smooth manifold of dimension $d$ and let $f \colon M \to M$ be a $C^1$-map. Let $p_0 \in M$ be a hyperbolic fixed point of $f$ and let $W^{s,u}_{loc}(p_0)$ be the local stable and unstable manifolds associated to $p_0$, respectively. The point $p_0$ is a \emph{generalized snap-back repeller (GSPR)} if:
\begin{itemize}
\item $u = \dim E^u > 0$;
\item $\exists$ a point $q_0 \in W^u_{loc}(p_0)$, $q_0 \ne p_0$ and $\exists k \ge 1$ such that $f^{(k)}(q_0) \in W^u_{loc}(p_0)$;
\item there exists a $u$-dimensional disk $B^u$ embedded in $W^u_{loc}(p_0)$ s.t. $B^u$ is a neighborhood of $q_0$, $f^k_{|B^u} \colon B^u \to M$ is an embedding, and $f^k(B^u)$ intersects $W^s_{loc}(p_0)$ transversally at $f^k(q_0)$.
\end{itemize}
\end{defi}

\begin{theo}[Shiraiwa-Kurata, 1981]\label{Shiraiwa_theo}
If a discrete map $f$ has a GSPR, then it is chaotic in the sense of Li-Yorke.
\end{theo}

\begin{rem}\label{rem_GSPR}
As for the SPR, the existence of a GSPR implies Devaney chaos when the map is restricted to a Cantor set.
\end{rem}

Now, we state the main theorem proving that, under suitable conditions, the cross-coupling \eqref{mapC} preserves chaos in the sense of Devaney if it is associated to the uncoupled system \eqref{mapR}. In particular, we show that if the Rulkov map \eqref{mapR} has a SPR, then the cross-coupling \eqref{mapC} of two identical Rulkov maps has a GSPR that can be explicitly constructed starting from the SPR associated to \eqref{mapR}.

In the following, ``returning point'' stands for the point $q_0$ associated to the fixed point $p_0$ as in Definition \ref{Marotto_defi} - \ref{Shiraiwa_defi}.

\begin{theo}\label{theo_cross}
Consider the systems \eqref{mapR} and \eqref{mapC}, with $f(x) = \frac{1}{1 + x^2}$, $\alpha = \beta > 0$, $\mu = \nu \in (0, 1)$, $\sigma = \rho \in \bb{R}$. Suppose that $p_0$ is a SPR of order $2$ for \eqref{mapR}, the returning point $q_0$ is sufficiently close to $p_0$ and $\frac{2 \alpha \sigma}{ (1+\sigma^2)^2 } \ne - 1 - \frac{\mu}{2}$. Then, \eqref{mapC} has a GSPR, up to a set of singular parameters.
\end{theo}
\begin{proof}
We name $R$ the Rulkov map \eqref{mapR} defined by parameters $(\alpha, \mu, \sigma)$ and we name $C$ the cross-coupling \eqref{mapC} of $R$ with itself. In \textit{Step 1} we prove the first and the second condition required by Definition \ref{Shiraiwa_defi} for a general choice on $f \in C^1(\bb{R})$ and returning step $k \ge 2$, while the third condition is proved in \textit{Step 2} by taking $f(x) = \frac{1}{1 + x^2}$ and $k=2$.
\begin{enumerate}[label=\textit{Step \arabic*:}]
\item Let $p_0 = (\sigma, \sigma - \alpha f(\sigma) )$ be the SPR of $R$ associated to the repelling neighborhood $B$ and on the returning point $q_0$ of order $k$, i.e. $R^{(k)}(q_0) = p_0$.
\\The fixed point of $C$ is $\tilde p_0 := (p_0, p_0)$. Since $p_0$ is a hyperbolic repelling point for \eqref{mapR} and the map \eqref{mapR} is $C^1$, there exist an unstable manifold $V^u$ associated to $p_0$ and a repelling neighborhood $B \subset V^u$. Since the point $\tilde p_0$ is a hyperbolic saddle point (see \textit{Step 2}) and the map \eqref{mapC} is $C^1$, by the stable manifold theorem \cite{Teschl2012} there exist a stable manifold $W^s$ and an unstable manifold $W^u$ associated to $\tilde p_0$, with $\tilde p_0 \in W^s \cap W^u$.
\\By hypothesis, the point $q_0$ is such that
\beq
\lim_{n+\infty} R^{(-n)}(q_0) = \lim_{n \to +\infty}
\begin{pmatrix}
\alpha f(q_{0x}) + q_{0y} \\ q_{0y} - \mu (q_{0x} - \sigma)
\end{pmatrix} =
\begin{pmatrix}
p_{0x} \\ p_{0y}
\end{pmatrix} = p_0 \,.
\eeq
Therefore,
\beq\begin{split}
\lim_{n+\infty} C^{(-n)}(q_0, q_0) &= \lim_{n \to +\infty}
\begin{pmatrix}
\alpha f(q_{0z}) + q_{0y} \\
q_{0y} - \mu (q_{0x} - \sigma) \\
\alpha f(q_{0x}) + q_{0\omega}
\\ q_{0\omega} - \mu (q_{0z} - \sigma)
\end{pmatrix} = \\
&= \lim_{n \to +\infty}
\begin{pmatrix}
\alpha f(q_{0x}) + q_{0y} \\
q_{0y} - \mu (q_{0x} - \sigma) \\
\alpha f(q_{0x}) + q_{0y}
\\ q_{0y} - \mu (q_{0x} - \sigma)
\end{pmatrix} =
\begin{pmatrix}
p_{0x} \\ p_{0y} \\ p_{0x} \\ p_{0y}
\end{pmatrix} = \tilde p_0 \,.
\end{split}\eeq
Consequently, $q_0 \in V^u$ $\implies$ $\tilde q_0 := (q_0, q_0) \in W^u$.
\\By construction, we have
\beq\label{mapC_return}
\begin{cases}
C_1^{(k)}( q_{0z} \,, q_{0y} ) = p_{0x} \\
C_2^{(k)}( q_{0x} \,, q_{0y} ) = p_{0y} \\
C_3^{(k)}( q_{0x} \,, q_{0\omega} ) = p_{0z} = p_{0x} \\
C_4^{(k)}( q_{0z} \,, q_{0\omega} ) = p_{0\omega} = p_{0y} \,,
\end{cases}
\eeq
that is $C^{(k)}( \tilde q_0) = \tilde p_0$.
\\In order to prove that $C^{(k)}$ is an embedding when the domain is restricted to a proper disk $B^u(\tilde q_0) \in W^u$, it is sufficient to verify that $\det D C^{(j)}(\tilde q_0) \ne 0$, $j = 1, \dots, k$, so that $D C^{(k)}(\tilde q_0)$ is locally diffeomorphic. We name $g = f'$. Since by hypothesis $R$ has a SPR at $q_0$, then $R$ is diffeomorphic at $q_0$ for every iteration $j = 1, \dots, k$, that is
\beq
\det D R^{(k)}(q_0) = \alpha g^{(j)}( q_{0x} ) + \mu \ne 0 \quad \forall j = 1, \dots, k \,.
\eeq
Now, we have
\beq\begin{split}
\det D C^{(k)}(\tilde q_0) &= -\alpha^2 ( g^{(j)}( \tilde q_{0x} ) )^2 + \mu^2 = \\
&= -\alpha^2 ( g^{(j)}( q_{0x} ) )^2 + \mu^2 = \\
&= [ \mu + \alpha g^{(j)}( q_{0x} ) ] \cdot [ \mu - \alpha g^{(j)}( q_{0x} ) ] \,.
\end{split}\eeq
Since $R$ has a SPR, the orbit of $q_0$ lies inside a repelling neighborhood of $p_0$ and $\forall j = 1, \dots, k$ it holds true the repelling condition
\beq
\alpha g^{(j)}(q_0) < \mu - 1 \implies \alpha g^{(j)}(q_0) < 0 \,,
\eeq
since $\mu \in (0, 1)$. Then, $\mu - \alpha g^{(j)}( q_{0x} ) \ne 0$, hence $\det( D C^{(j)}(\tilde q_0) ) \ne 0$ $\forall j = 1, \dots, k$.
\item Now, we prove the transversal intersection between the stable and unstable manifolds associated to $\tilde p_0$, as requested by Definition \ref{Shiraiwa_defi}. For the rest of the proof, we assume $f(x) = \frac{1}{1 + x^2}$ and $k=2$.
\\The Jacobian of \eqref{mapC} is given by
\beq
D C =
\begin{pmatrix}
0 & 1 & -K(z) & 0 \\
-\mu & 1 & 0 & 0 \\
-K(x) & 0 & 0 & 1 \\
0 & 0 & -\mu & 1
\end{pmatrix} \,, \quad K(t) = \frac{2\alpha t}{ (1+t^2)^2 } \,.
\eeq
We set $K := K(\sigma) = \frac{2\alpha \sigma}{ (1 + \sigma^2)^2 }$. An explicit calculation of the eigenvalues of $D C$ at $\tilde p_0$ returns
\beq\begin{split}
\lambda_{1,2} = \frac{1 - K}{2} \mp \frac{ \sqrt{ (K+1)^2 - 4\mu } }{2} \,, \\
\lambda_{3,4} = \frac{1 + K}{2} \mp \frac{ \sqrt{ (K-1)^2 - 4\mu } }{2} \,. \\
\end{split}\eeq
The repelling condition on the map \eqref{mapR} becomes $K < \mu - 1$. By analyzing the signs of the terms $(K+1)^2 - 4\mu$, $(K-1)^2 - 4\mu$, we conclude that $\forall K < \mu - 1$ $\forall \mu \in (0, 1)$
\beq
|\lambda_{1,2}| > 1 \,, \quad |\lambda_4| < 1 \,,
\eeq
while
\beq
\begin{aligned}
|\lambda_3| &> 1
    && \text{if} && K < -1 - \frac{\mu}{2} \,, \\[0.5ex]
|\lambda_3| &= 1
    && \text{if} && K = -1 - \frac{\mu}{2} \,, \\[0.5ex]
|\lambda_3| &< 1
    && \text{if} && -1 - \frac{\mu}{2} < K < \mu - 1 \,.
\end{aligned}
\eeq
Therefore, for $K \ne -1 - \frac{\mu}{2}$ the point $\tilde p_0$ is a hyperbolic saddle, as stated in \textit{Step 1}. In particular, $\dim E^u = 3$ and $\dim E^s = 1$ if $K < -1 - \frac{\mu}{2}$, $\dim E^u = \dim E^s = 2$ if $-1 - \frac{\mu}{2} < K < \mu - 1$, where $E^{s, u}$ are the stable and the unstable eigenspace associated to $\tilde p_0$, respectively.
\\We name $W^{s, u} := W^{s, u}(\tilde p_0)$ the stable and unstable manifolds associated to $\tilde p_0$, respectively. We want to prove that $W^s \pitchfork W^u$ at $\tilde p_0$, i.e.
\beq
T_{\tilde p_0} ( W^s ) \cap T_{\tilde p_0} ( C^{(k)} (W^u) ) = \{ 0 \} \,,
\eeq
where $T_{\tilde p_0}$ denotes the tangent space to the manifold at the point $\tilde p_0$. By the stable manifold theorem, $T_{\tilde p_0} ( W^s ) = E^s$ and we can replace $T_{\tilde p_0} [ C^{(k)} (W^u) ]$ with the tangent space at $\tilde q_0$ propagating the dynamics through the Jacobian of \eqref{mapC}, leading to
\beq
E^s \cap D C^{(k)}(\tilde q_0) T_{\tilde q_0}(W^u) = \{ 0 \} \,.
\eeq
Now, we assume that $q_0$ is sufficiently close to $p_0$, i.e. $\tilde q_0$ sufficiently close to $\tilde p_0$, such that we can approximate the tangent space of $W^u$ at $\tilde q_0$ with its tangent space at $\tilde p_0$, i.e. with the unstable eigenspace $E^u$. Therefore, we are left to prove that
\beq\label{eq_M}
E^s \cap D C^{(k)}(\tilde q_0) E^u = \{ 0 \} \,,
\eeq
that is
\beq\label{eq_rank}
\rank
\begin{pmatrix}
E^s & | & D C^{(k)}(\tilde q_0) E^u
\end{pmatrix} \ne 0 \,.
\eeq
By performing symbolic computations on \eqref{mapC} (see \ref{sec_app} for the details), we conclude that $W^u \pitchfork W^s$, hence \eqref{mapC} has a GSPR, $\forall K < \mu - 1$ with $K \ne -1 - \frac{\mu}{2}$, up to the set of values $(\alpha, \mu, \sigma)$ that do not satisfy \eqref{eq_rank}.
\end{enumerate}
\end{proof}

\begin{rem}
The hypothesis of Theorem \eqref{theo_cross} regarding the closeness of $q_0$ and $p_0$ is generally valid, in view of the results established in \cite{GeCao2021}. In particular, for $\alpha = 15$, $\mu = 0.001$, $\sigma \in [-2, -0.1]$, the existence of a snap-back repeller for \eqref{mapR} proved in \cite{GeCao2021} implies the existence of a GSPR for \eqref{mapC}.
\end{rem}

\begin{cor}\label{cor_cross}
Suppose that the systems \eqref{mapR} and \eqref{mapC} satisfy Theorem \ref{theo_cross}. Then, \eqref{mapC} has a GSPR for sufficiently small values on $\alpha - \beta$, $\mu - \nu$, $\sigma - \rho$.
\end{cor}
\begin{proof}
If $\alpha \ne \beta$, $\mu \ne \nu$, $\sigma \ne \rho$, \eqref{mapC} can be written as the case analyzed in the proof of Theorem \ref{theo_cross} plus a perturbation of class $C^1$ on the variables $y_n$, $\omega_n$. The persistence of a SPR under continuously differentiable small perturbations in Banach spaces has been proved in \cite{ChenHuang2011}, and this proof can be generally extended to the case of a GSPR.
\end{proof}

In the classical regime on the parameters, i.e. $0 < \mu, \nu \ll 1$, the biological interpretation of Theorem \ref{theo_cross} and Corollary \ref{cor_cross} is immediate: the cross-coupling preserves the emergence of Devaney chaos when it is associated to two Rulkov neurons that are defined through the equal or similar membrane potentials, perturbations on the slow variable and external currents.

Even though our approach focused on the analytical definition of the Rulkov map, the arguments in the proof of Theorem \ref{theo_cross} may be adapted to other discrete maps of biological interest also coming from different frameworks, such as proper reformulations of the Guevara map \cite{Guevara1982} for cardiac arrhythmia or the Chialvo map \cite{Chialvo1995} for the neuron.

\section{Numerical studies}\label{sec_num}
In this section we focus on the cross-coupling of two chaotic Rulkov maps taking $f(x) = \frac{1}{1 + x^2}$ in \eqref{mapC}, that is
\beq\begin{cases}\label{Rulkov4D}
x_{n+1} = \frac{\alpha}{1 + z_n^2} + y_n \\
y_{n+1} = y_n - \mu(x_n - \sigma) \\
z_{n+1} = \frac{\beta}{1 + x_n^2} + \omega_n \\
\omega_{n+1} = \omega_n - \nu(z_n - \rho) \,.
\end{cases}\eeq
We present several numerical evidences for the emergence of chaotic behavior for the cross-coupling \eqref{Rulkov4D}, under general choices on the parameters.

For all the following numerical simulations we always take random initial conditions $x_0, y_0, z_0, \omega_0 \in [0, 1]$ and only specify the values assumed on the parameters of the system. In the following, ``MLE'' stands for maximum Lyapunov exponent.

\subsection{A potential strange attractor}\label{sec_num_attractor}
We identify for a wide sets of parameters the arising of a potential strange attractor, that is present in all 2D and 3D projection planes of the phase space. A representative example on the plane $(x, z)$ is shown in Figure \ref{attractor_tot}, where the first $10 \%$ of transient points are discarded to better visualize the final attractor reached by the system.

\begin{figure}[h!]
\centering
\subfloat[][\label{attractor}]
{\includegraphics[width=.45\textwidth]{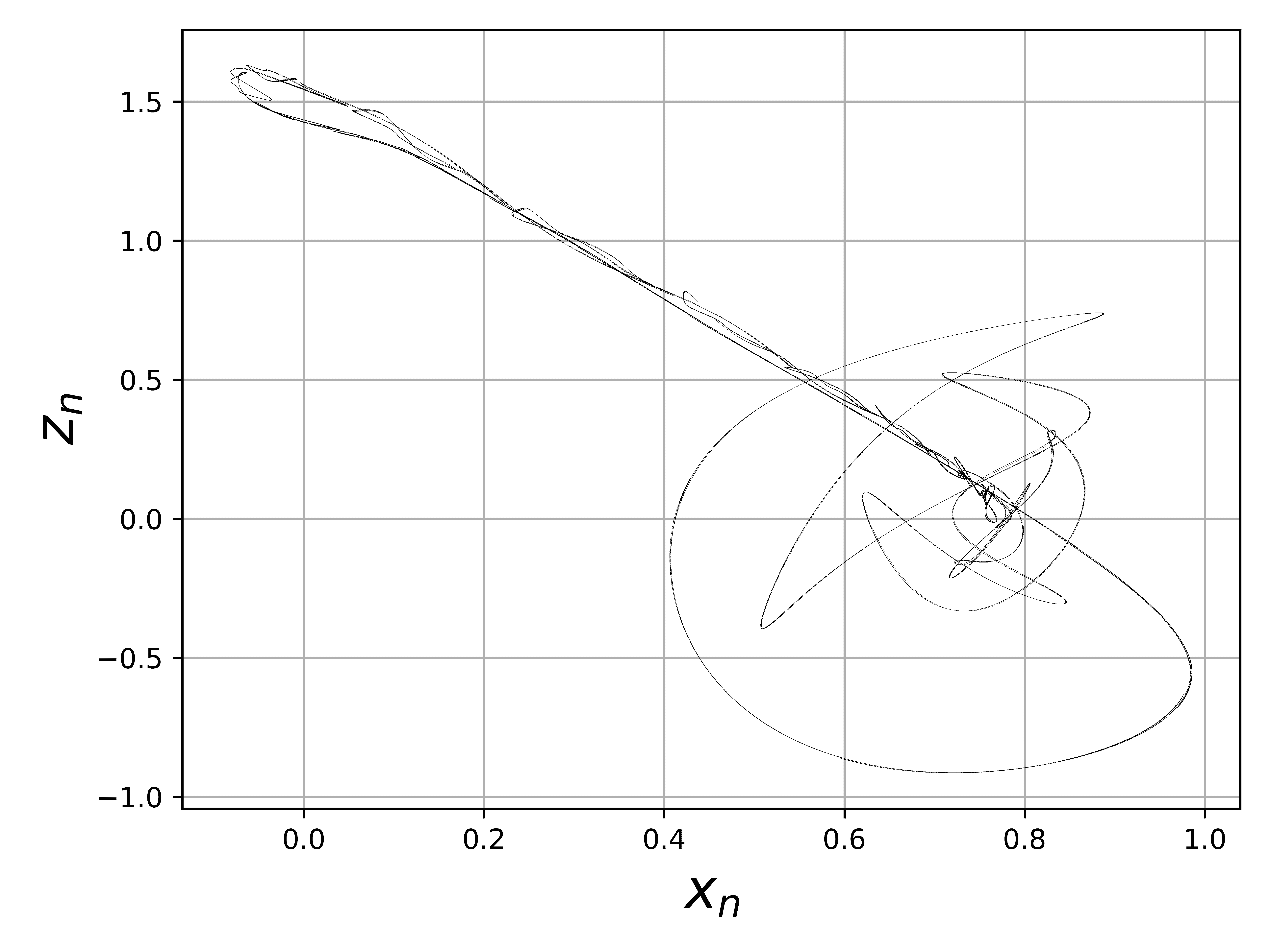}} \quad
\subfloat[][\label{attractor_zoom}]
{\includegraphics[width=.45\textwidth]{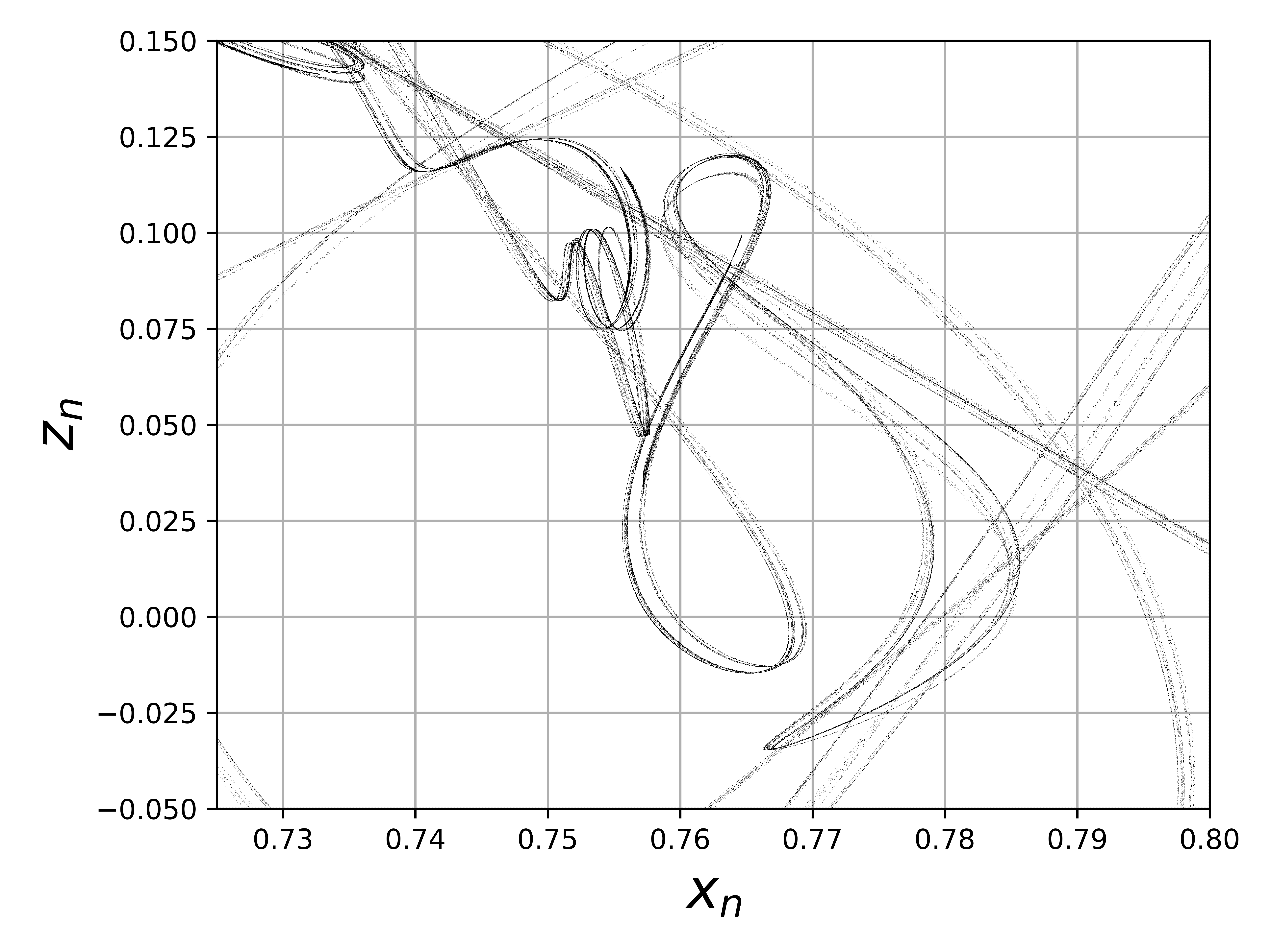}}
\caption{a) Potential strange attractor for the system \eqref{Rulkov4D}. Parameters: $\alpha = 1.1$, $\beta = 4$, $\mu = \nu = 0.1$, $\sigma = \rho = 0.5$. The simulation run for $1.000.000$ iterations. b) Magnification of a) for $x \in [0.725, 0.8]$, $z \in [-0.05, 0.15]$. The simulation run for $10.000.000$ iterations.}
\label{attractor_tot}
\end{figure}

The strangeness of the attractor is strongly suggested by the computation of the Kaplan-Yorke dimension \cite{Kaplan1979}.

\begin{defi}[Kaplan-Yorke, 1979]
Suppose that the dynamical system $f \colon \bb{R}^n \to \bb{R}^n$ has an attractor and let $\lambda_1 \ge \lambda_2 \ge \dots \ge \lambda_n$ be its Lyapunov exponents. Let $j \in \bb{N}$ be the largest positive integer such that
\[
\sum_{i=1}^j \lambda_i \ge 0 \,, \quad \sum_{i=1}^{j+1} \lambda_i < 0 \,.
\]
The \emph{Kaplan-Yorke dimension} of the attractor is defined as
\[
d_{KY} = j + \frac{ \sum_{i=1}^j \lambda_i }{ |\lambda_{j+1}| } \,.
\]
\end{defi}

The original Kaplan-Yorke conjecture stated that $d_{KY}$ coincides with the Hausdorff dimension for any attractor. Even though it has been proved that this is not generally true, its computation has become the standard in the numerical study of chaotic systems, since it provides an extremely accurate estimate of the Hausdorff dimension, especially for high dimension systems for which the computation of the Hausdorff dimension via box-counting is often prohibitive.

In Figure \ref{kaplan_dim} the Kaplan-Yorke dimension of \eqref{Rulkov4D} is computed for $\beta = 4$, $\mu = \nu = 0.1$, $\sigma = \rho = 0.5$ and $500$ values for $\alpha \in [1, 2]$. We can argue that when $d_{KY}$ is relatively far from an integer value, then we may expect a non-integer Hausdorff dimension of the attractor. In particular, for the attractor shown in Figure \ref{attractor_tot} the Kaplan-Yorke dimension is $d_{KY} \in [1.79, 1.82]$.

\begin{figure}[h!]
\centering
\includegraphics[width=8.5cm]{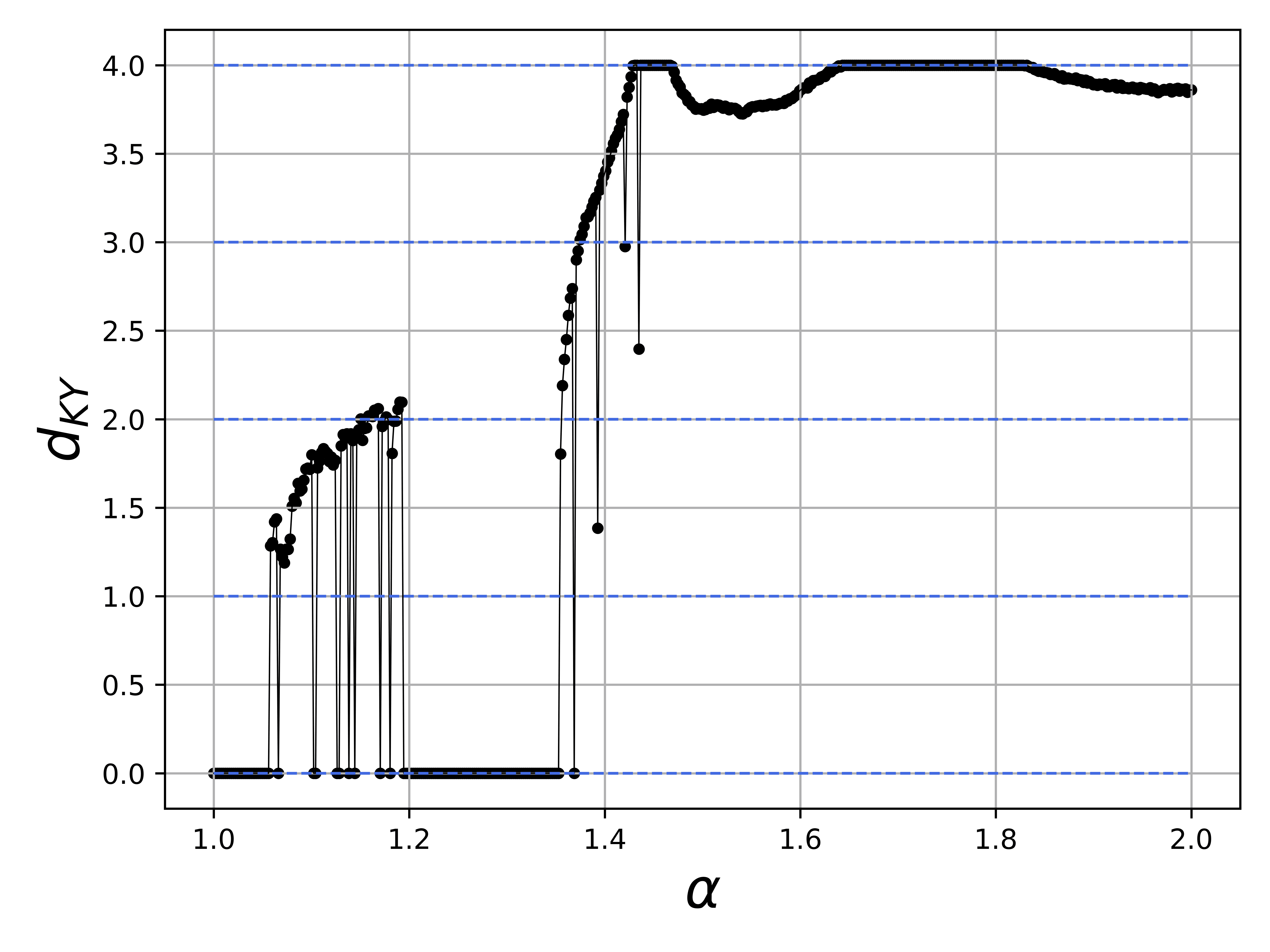}
\caption{Kaplan-Yorke dimension of the system \eqref{Rulkov4D} for $\alpha \in [1, 2]$ ($500$ values), $\beta = 4$, $\mu = \nu = 0.1$, $\sigma = \rho = 0.5$. Lyapunov exponents are computed under $1.000.000$ iterations starting from a set of random initial conditions $(x_0, y_0, z_0, \omega_0) \in (0, 1)^4$. The horizontal lines $d_{KY} = 0, \dots, 4$ are outlined in blue in order to better visualize when $d_{KY} \in \bb{N}$. The Kaplan-Yorke dimension associated to the attractor shown in Figure \ref{attractor_tot} is $d_{KY} \in [1.79, 1.82]$.}
\label{kaplan_dim}
\end{figure}

Values of the Kaplan-Yorke dimension generally agree with the appearing and disappearing of the attractor as $\alpha$ varies. As a pair of examples, in Figure \ref{kaplan_example1} they are shown the time series $x_n$, $z_n$ for $\alpha = 1.3$ and associated $d_{KY} = 0$, i.e. the attractor is a set of measure zero, while in Figure \ref{kaplan_example2} they are shown the time series $x_n$, $z_n$ for $\alpha = 1.7$ and associated $d_{KY} = 4$, i.e. the motion is essentially ergodic in a subset of $\bb{R}^4$. This last feature can be appreciated in Figure \ref{kaplan_example2_xz}, where any fractal structure encountered in the previous case is completely vanished.

\begin{figure}[h!]
\centering
\includegraphics[width=8.5cm]{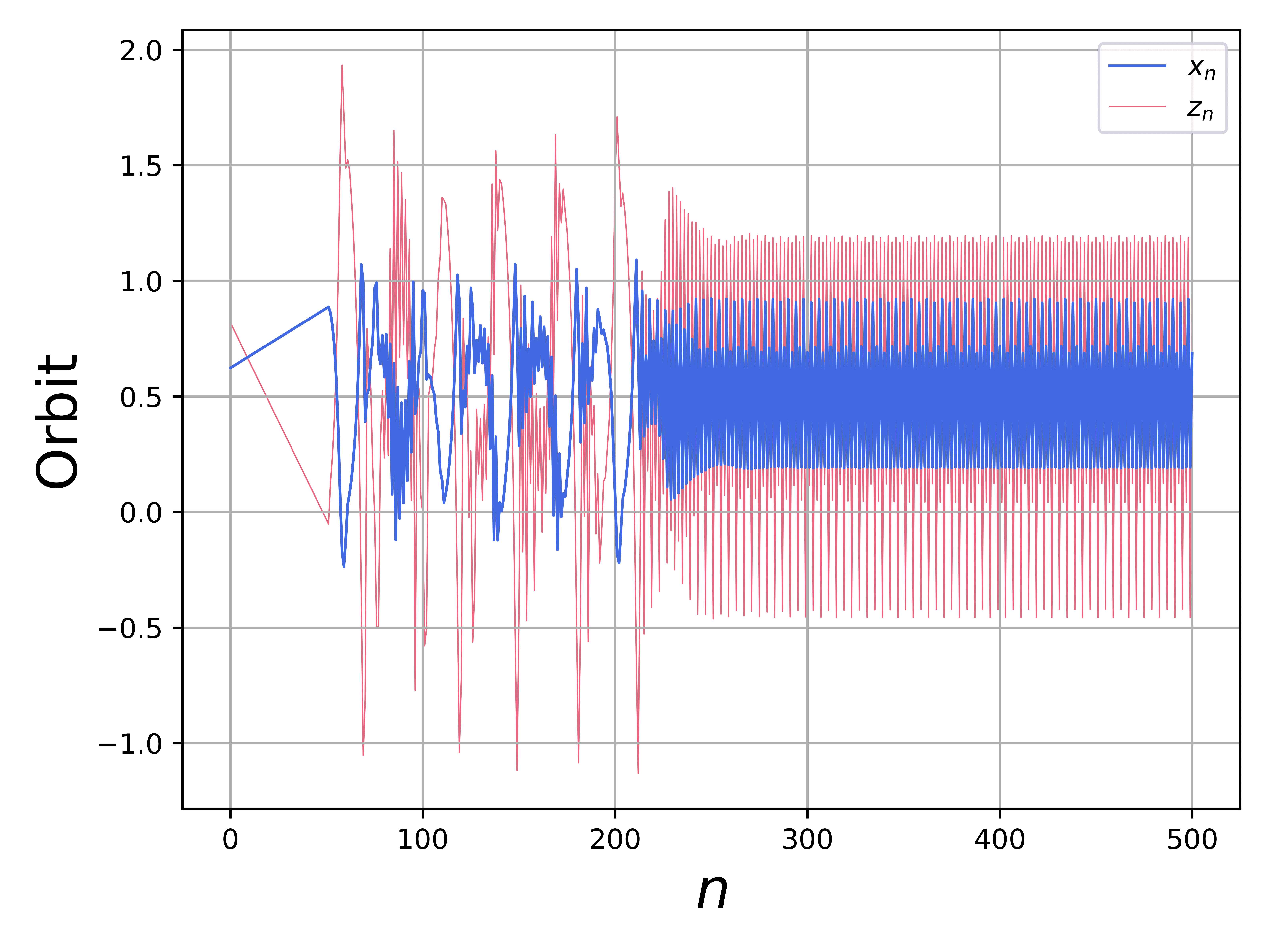}
\caption{Periodic time series reached by \eqref{Rulkov4D} after transient chaos for $\alpha = 1.3$ and Kaplan-Yorke dimension $d_{KY} = 0$. Other parameters: $\beta = 4$, $\mu = \nu = 0.1$, $\sigma = \rho = 0.5$. The simulation run for $500$ iterations.}
\label{kaplan_example1}
\end{figure}

\begin{figure}[h!]
\centering
\subfloat[][\label{kaplan_example2}]
{\includegraphics[width=.45\textwidth]{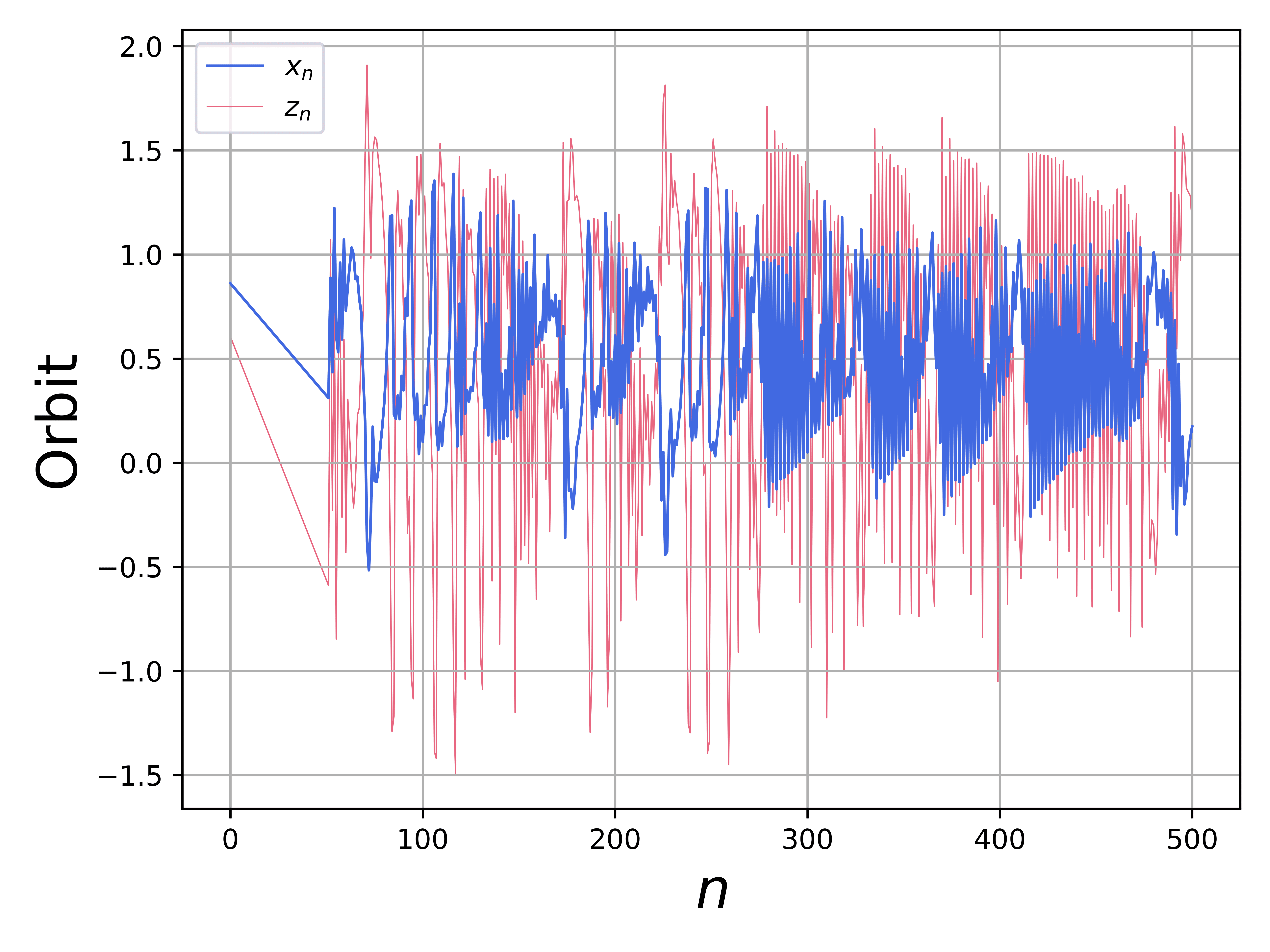}} \quad
\subfloat[][\label{kaplan_example2_xz}]
{\includegraphics[width=.45\textwidth]{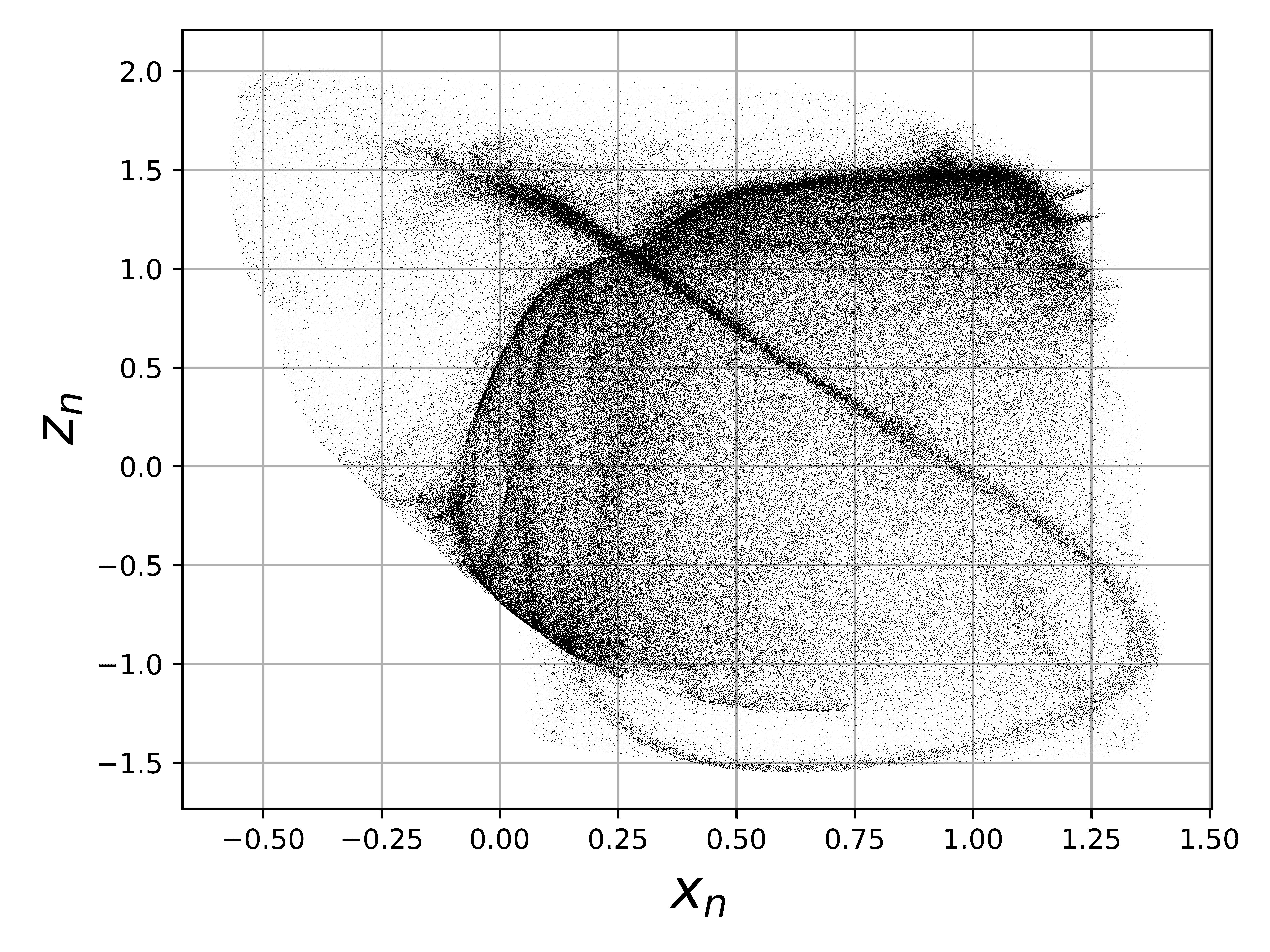}}
\caption{a) Chaotic time series of \eqref{Rulkov4D} for $\alpha = 1.7$ and Kaplan-Yorke dimension $d_{KY} = 4$. Other parameters: $\beta = 4$, $\mu = \nu = 0.1$, $\sigma = \rho = 0.5$. The simulation run for $500$ iterations. b) Chaotic attractor of \eqref{Rulkov4D} for $\alpha = 1.7$, see a). The simulation run for $5.000.000$ iterations.}
\label{kaplan_example2_tot}
\end{figure}

It is worth to notice that, in the regime analyzed up to now, the coupling \eqref{Rulkov4D} does not exhibit the classical bursting-firing phenomenon, that is replaced by an alternation between more or less aggressive chaotic firing. Therefore, it is important to compare the previous results to analogous numerical simulations for the classical regime, that translates in our framework in taking both $\mu, \nu \ll 1$. Time series shown in Figure \ref{time_small_munu} qualitatively confirm the agreement with the classical bursting-firing exhibited by the Rulkov map for small perturbations. It is important to remark that the two membrane potentials are always taken such that one is below and the other one is above the chaotic threshold for the Rulkov map.

\begin{figure}[h!]
\centering
\includegraphics[width=8.5cm]{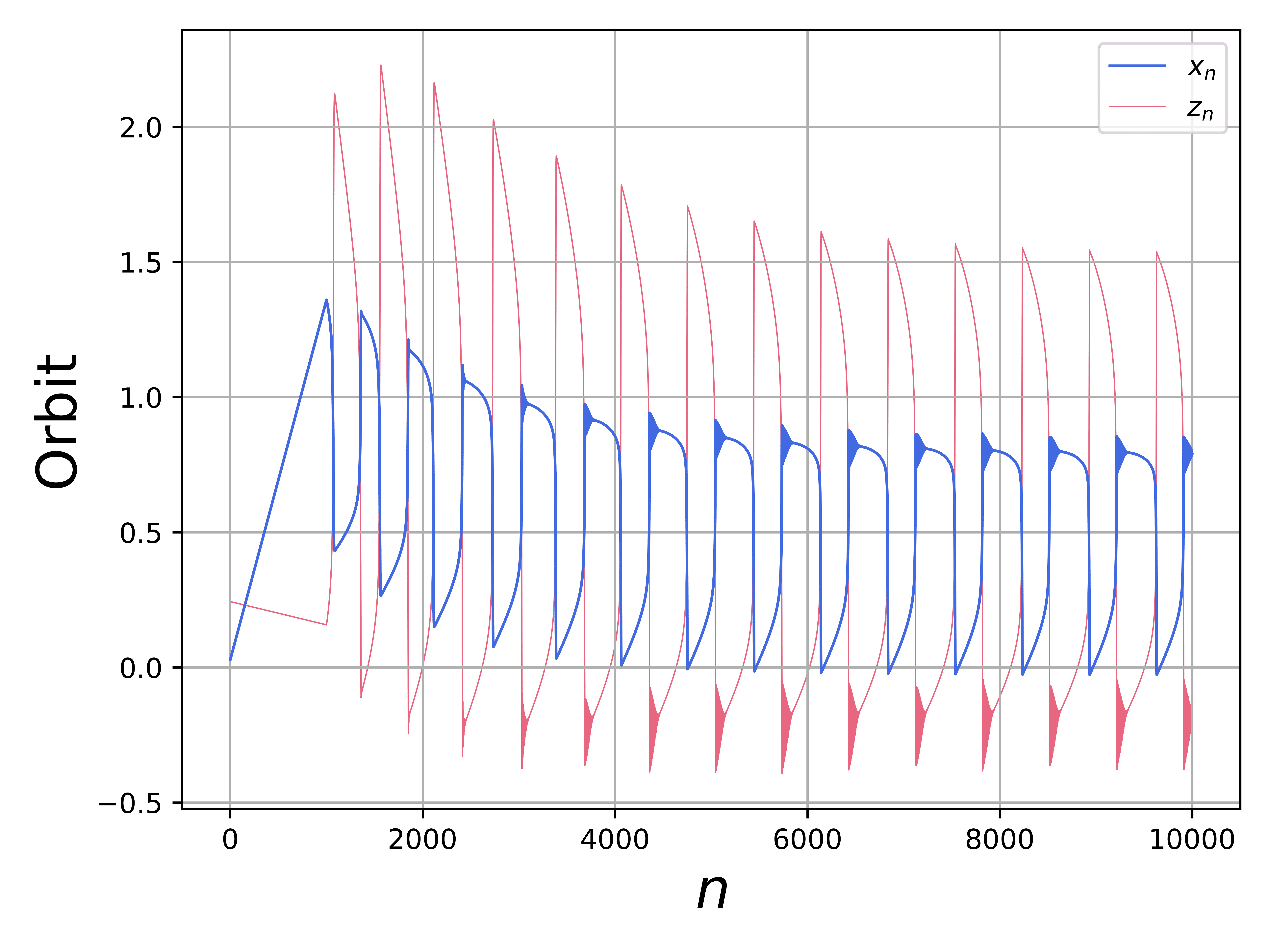}
\caption{Time series of \eqref{Rulkov4D} for $\mu, \nu \ll 1$. Here, we recognize the classical bursting-firing regime for both fast-variables $x, z$. Parameters: $\alpha = 1.1$, $\beta = 4$, $\sigma = \rho = 0.5$, $\mu = \nu = 0.001$. The simulation run for $10.000$ iterations.}
\label{time_small_munu}
\end{figure}

\subsection{Lyapunov exponents spectra, bifurcation diagrams, basins of attraction}\label{sec_num_lyap}
Figure \ref{lyap_alpha_b4m01s05} - \ref{lyap_mu_a11b4s05} show the Lyapunov exponents spectra of \eqref{Rulkov4D}. For these simulations, we adapt some of the several Python scripts presented in the exhaustive review \cite{LeGandhi2024}.

\begin{figure}[h!]
\centering
\includegraphics[width=8.5cm]{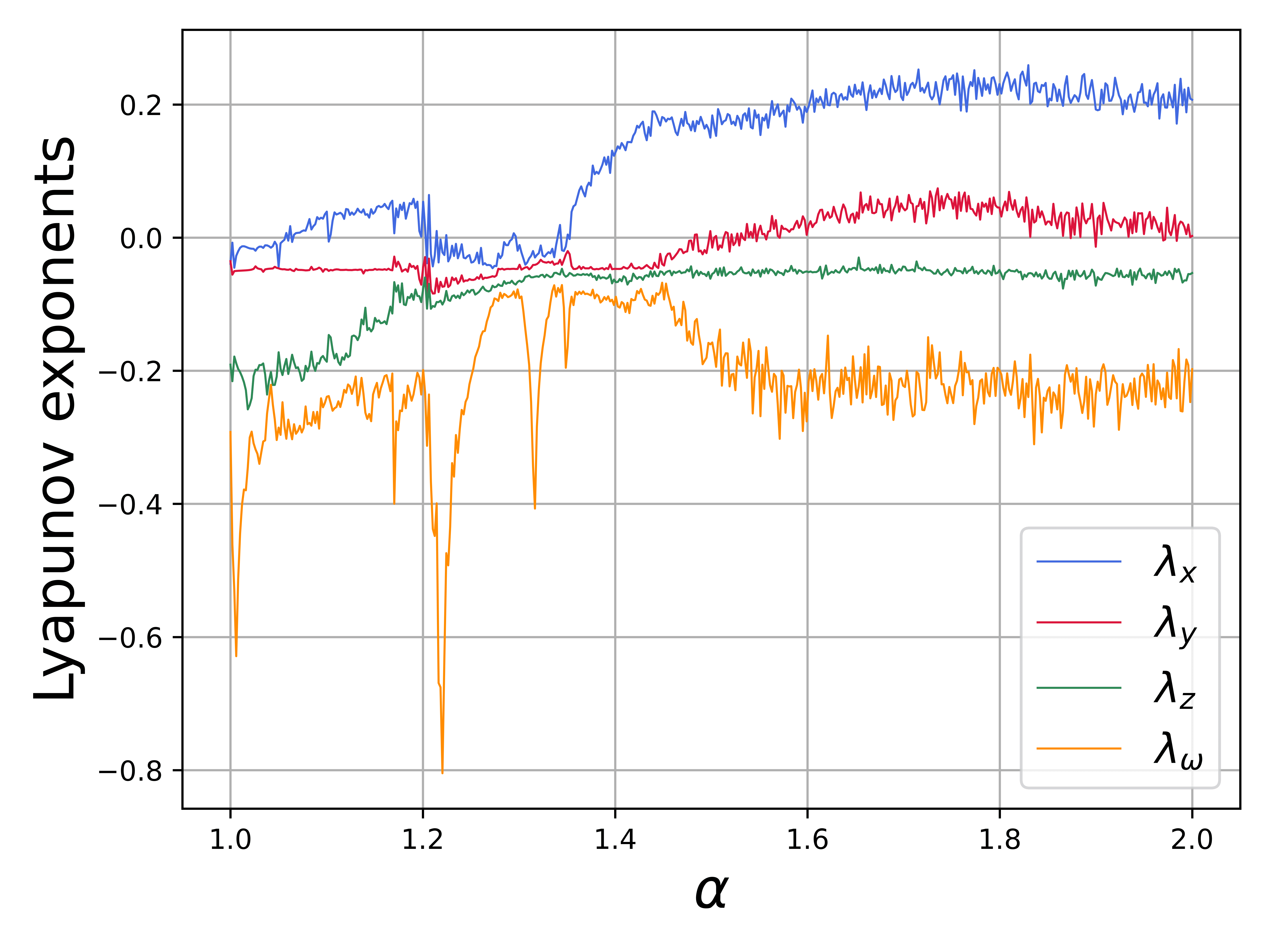}
\caption{Lyapunov exponents spectra of \eqref{Rulkov4D} for $\alpha \in (1, 2)$ ($500$ values), $\beta = 4$, $\mu = \nu = 0.1$, $\sigma = \rho = 0.5$. Lyapunov exponents are computed under $1.000$ iterations.}
\label{lyap_alpha_b4m01s05}
\end{figure}

\begin{figure}[h!]
\centering
\includegraphics[width=8.5cm]{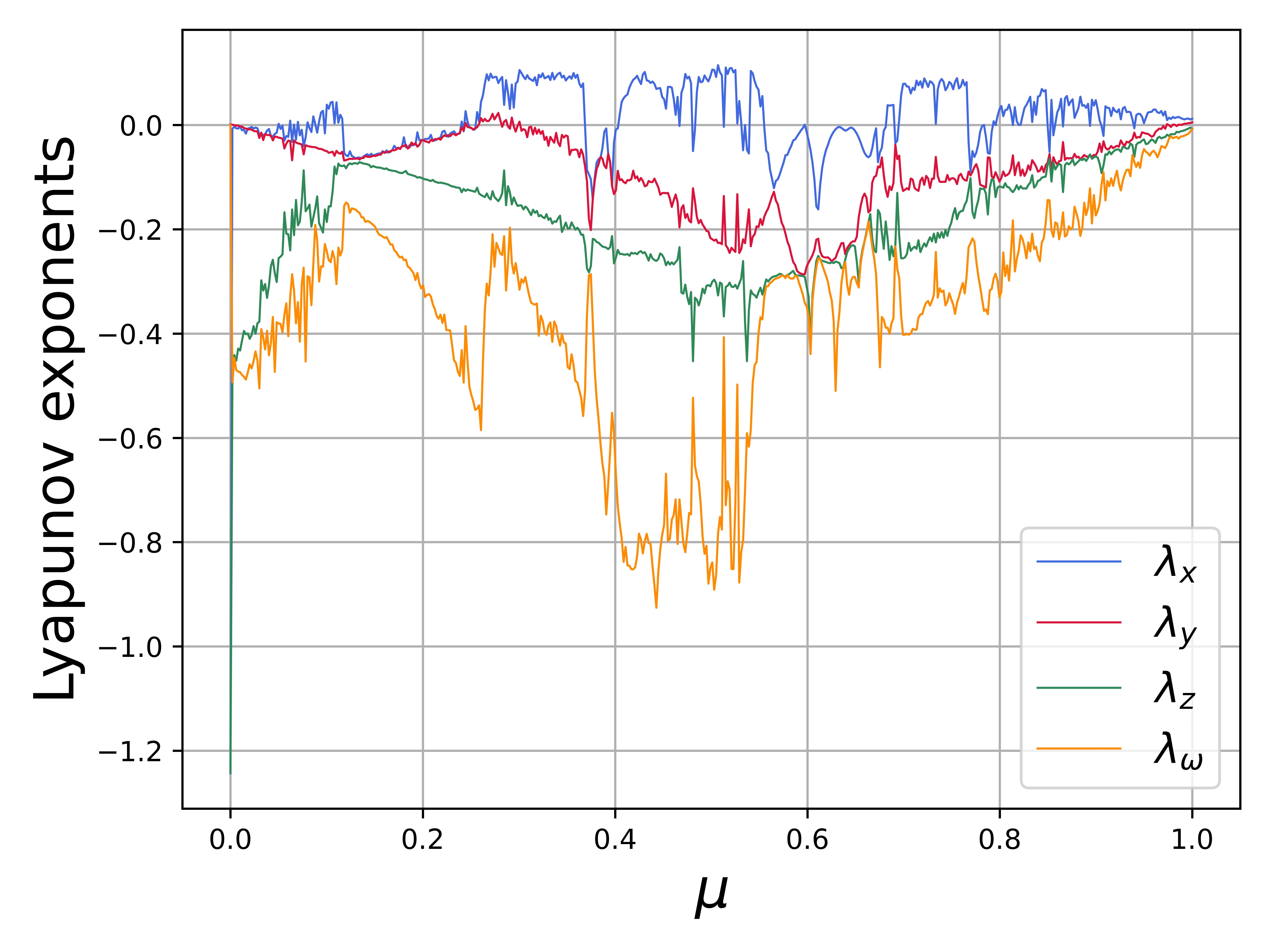}
\caption{Lyapunov exponents spectra of \eqref{Rulkov4D} for $\mu = \nu \in (0, 1)$ ($500$ values), $\alpha = 1.1$, $\beta = 4$, $\sigma = \rho = 0.5$. Lyapunov exponents are computed under $1.000$ iterations.}
\label{lyap_mu_a11b4s05}
\end{figure}

Firstly, it is worth to notice that in Figure \ref{lyap_mu_a11b4s05} all Lyapunov exponents approaches to zero when $\mu \to 1^-$ and this agrees with the general divergence for the motion of \eqref{Rulkov4D} when $\mu, \nu > 1$, that can be numerically checked.

Another property of the Lyapunov exponents spectra is easily explained. The parameter $\beta$ is the membrane potential associated to the subsystem $(z, \omega)$; however, given the formal structure of the cross-coupling, at each step $\beta$ influences the subsystem $(x, y)$. Therefore, it is clear that for $\beta = 4$, that is the threshold for the Rulkov map to exhibit chaos, the Lyapunov exponents of the subsystem $(x, y)$ are generally positive. For the same reason, both Lyapunov exponents of the subsystem $(z, \omega)$ tends to be negative since $\alpha = 1.1$ is lower than the chaotic threshold $\alpha = 4$. If we invert the values of $\alpha$, $\beta$, we end up to the opposite situation, with the exponents $\lambda_x$, $\lambda_y$ being generally negative and the exponents $\lambda_z, \lambda_\omega$ being generally positive.

Now, we perform a general bifurcation analysis of the system \eqref{Rulkov4D}, focusing on the parameters regime associated to the attractor shown in Figure \ref{attractor_tot}. In all the following simulations, the first $90 \%$ of transient points are discarded.

In Figure \ref{1Dbif_tot} we present 1D bifurcation diagrams for $\alpha \in (0.5, 1.5)$, $\beta = 4$, $\mu = \nu = 0.1$, $\sigma = \rho = 0.5$, showing for each state variable the arising of the usual Feigenbaum attractor for period doubling bifurcations together with islands of stability.

\begin{figure}[h!]
\centering
\subfloat[][\label{1Dbif_alfa_x_b4m01s05}]
{\includegraphics[width=.45\textwidth]{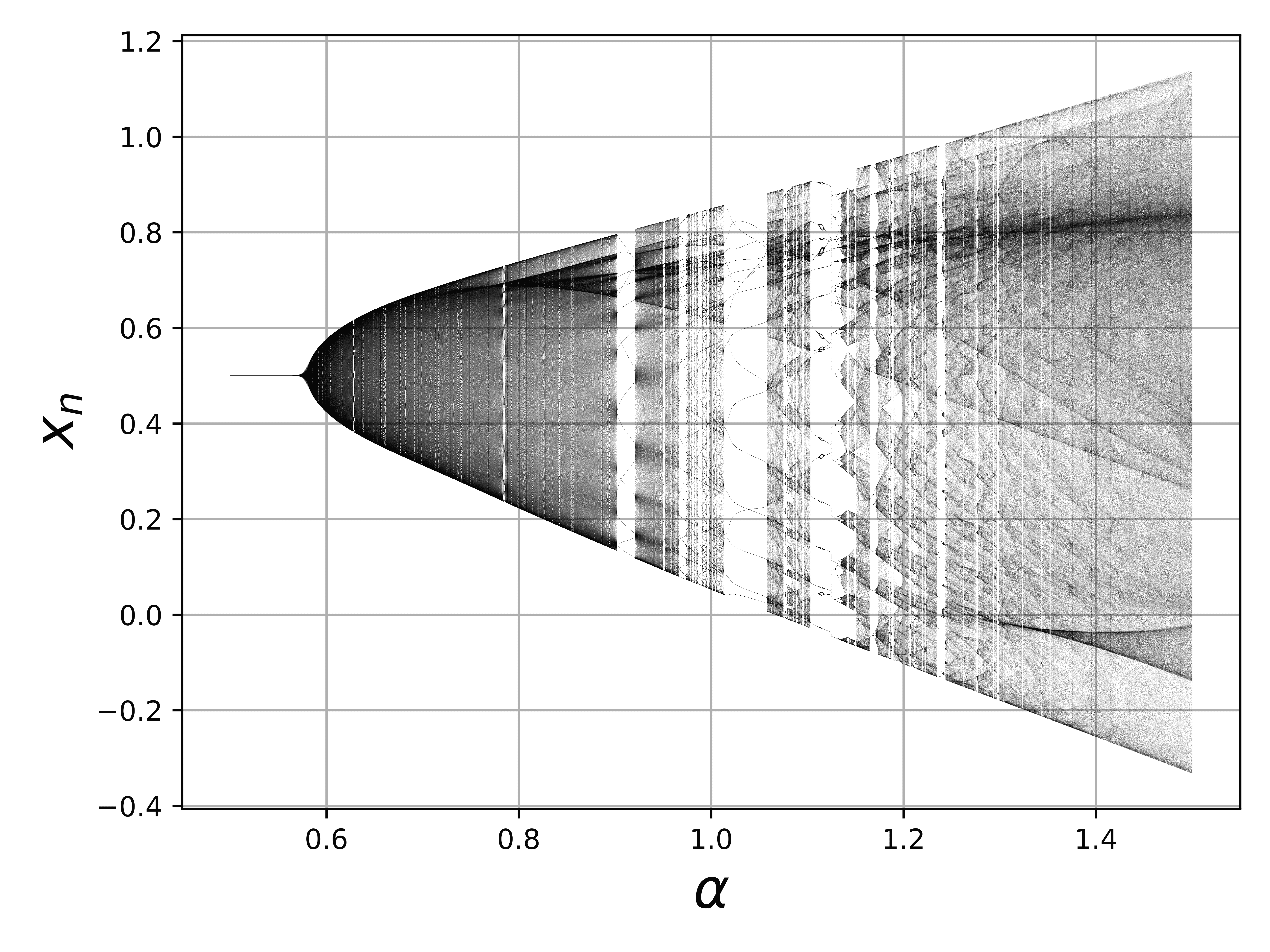}} \quad
\subfloat[][\label{1Dbif_alfa_y_b4m01s05}]
{\includegraphics[width=.45\textwidth]{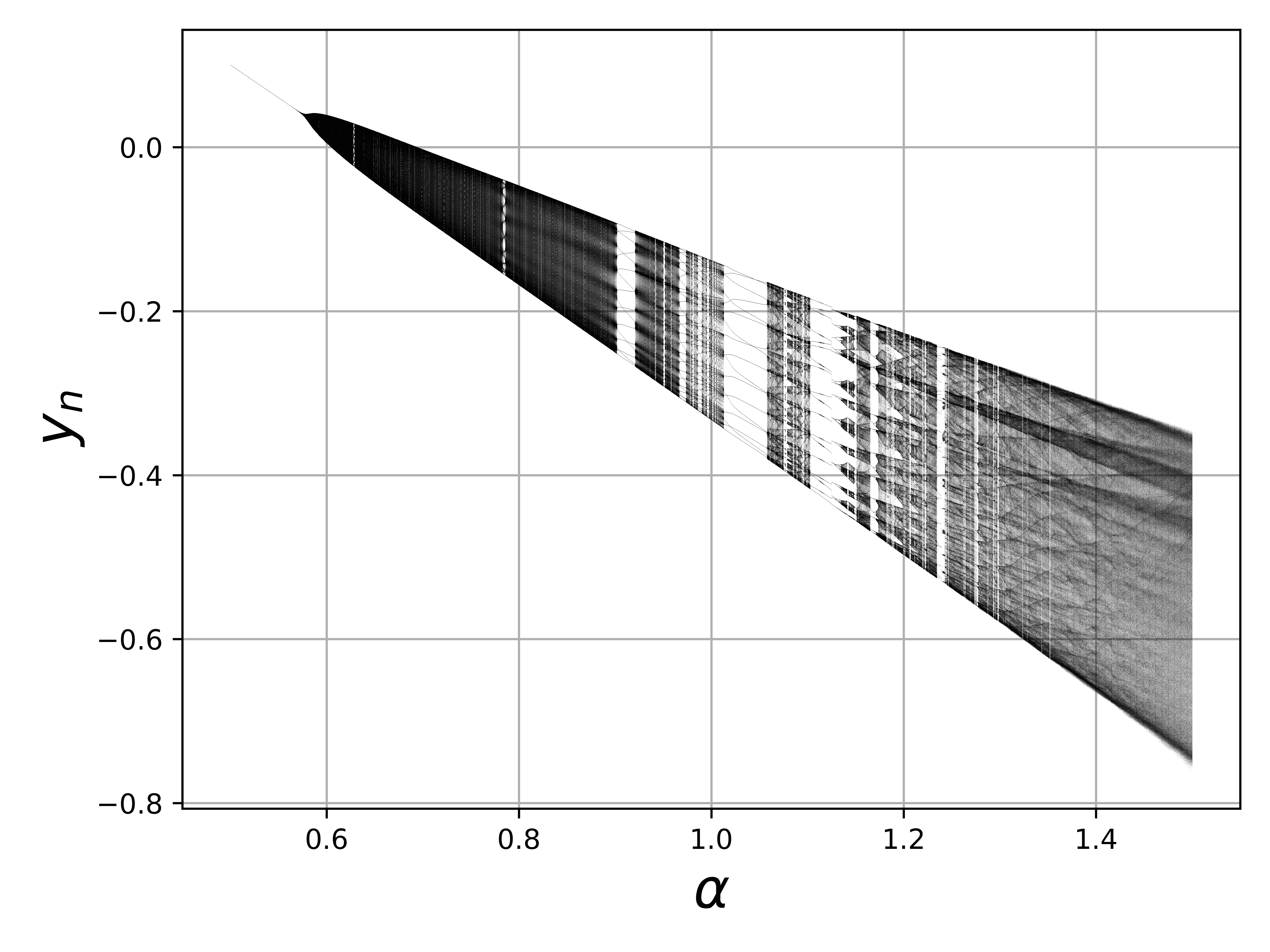}} \\
\subfloat[][\label{1Dbif_alfa_z_b4m01s05}]
{\includegraphics[width=.45\textwidth]{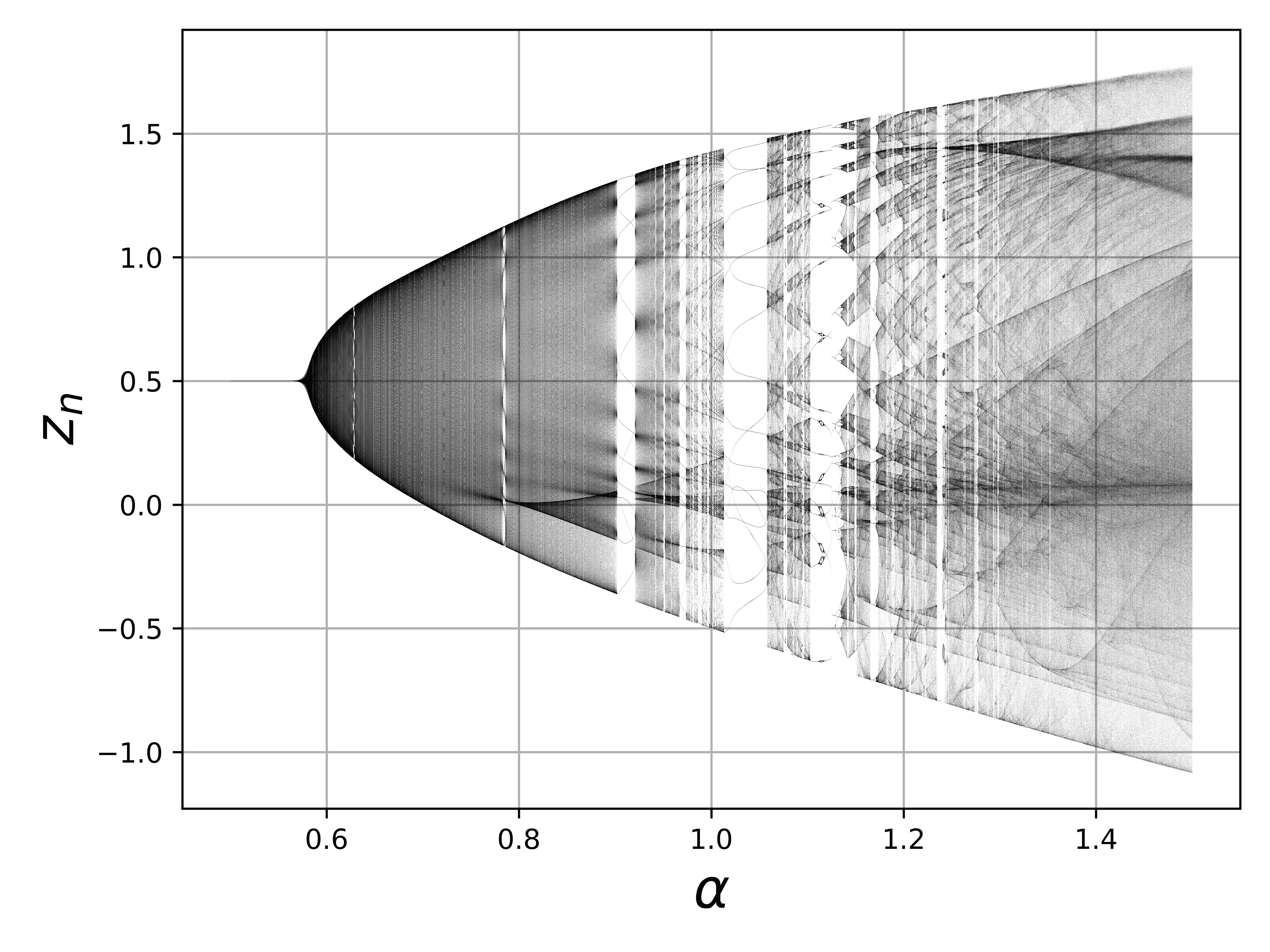}} \quad
\subfloat[][\label{1Dbif_alfa_w_b4m01s05}]
{\includegraphics[width=.45\textwidth]{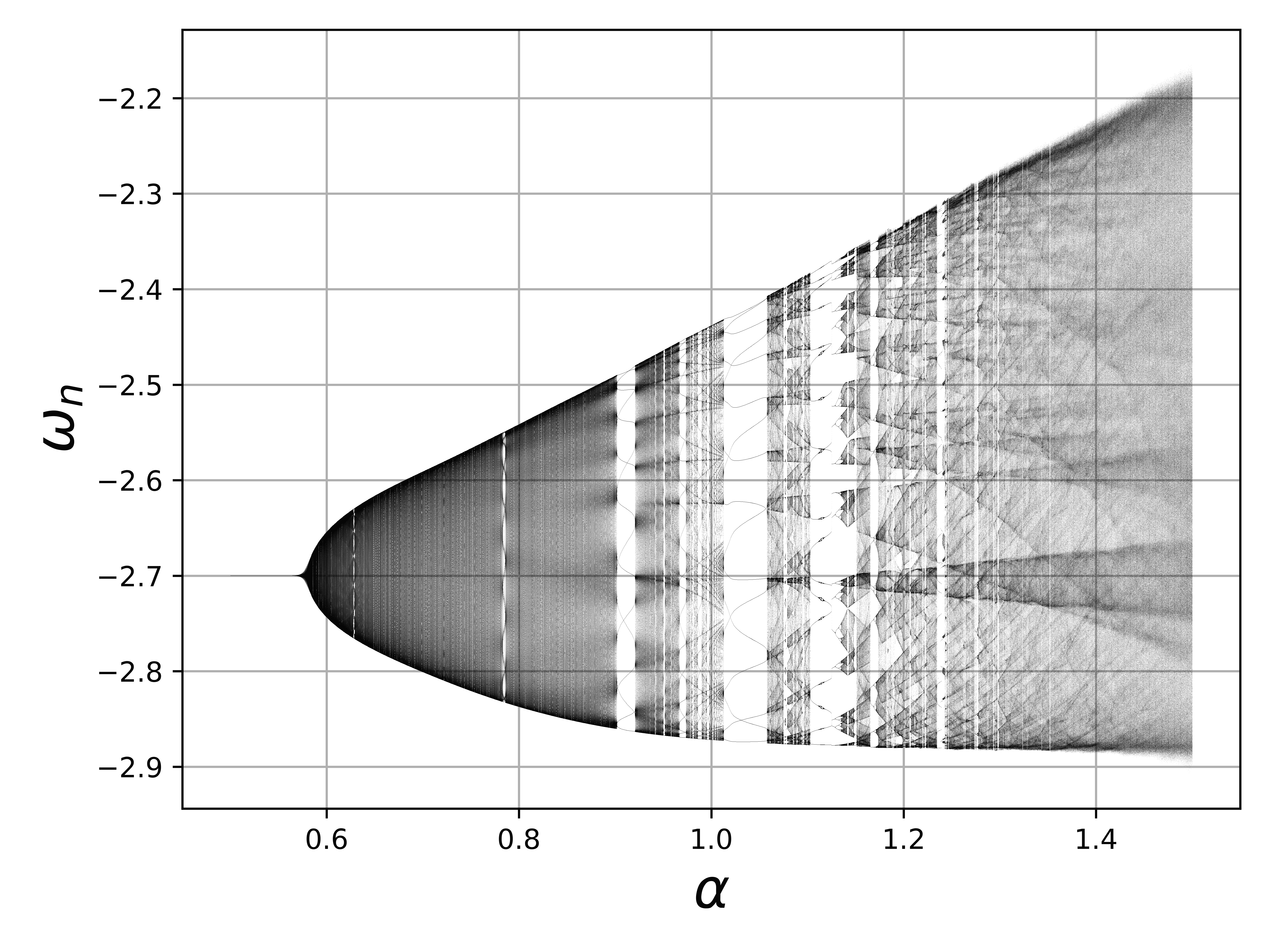}} \\
\caption{1D bifurcation diagrams for $\alpha \in (0.5, 1.5)$, $\beta = 4$, $\mu = \nu = 0.1$, $\sigma = \rho = 0.5$. The simulations run for $5.000$ iterations discarding the first $4.500$ transients. a) $(\alpha, x_n)$. b) $(\alpha, y_n)$. c) $(\alpha, z_n)$. d) $(\alpha, \omega_n)$.}
\label{1Dbif_tot}
\end{figure}

In Figure \ref{2Dbif_alfa_beta_m01s05} - \ref{2Dbif_sigma_rho_a11b4m01} we present 2D bifurcation diagrams for the system \eqref{Rulkov4D}. In these simulations, two variables vary on the horizontal and vertical axis, while the color identifies the magnitude of the MLE for a given pair of parameters. Choices on the fixed parameters are taken in order to compare the bifurcation diagrams one with the other and with the chaotic attractor shown in Figure \ref{attractor_tot}. The simulations implement the same Python script used for the spectra of Lyapunov exponents, returning the MLE after $100$ iterations on a $750 \times 750$ grid for the parameters.

In Figure \ref{2Dbif_alfa_beta_m01s05} it is presented the 2D bifurcation diagram $(\alpha, \beta)$ for $\alpha, \beta \in (1, 4)$, $\mu = \nu = 0.1$, $\sigma = \rho = 0.5$. Since the parameters $\alpha$, $\beta$ model the membrane potentials acting on the two neurons, Figure \ref{2Dbif_alfa_beta_m01s05} provides a visualization of the interaction between the two potentials for fixed perturbations and external currents.

\begin{figure}[h!]
\centering
\includegraphics[width=8.5cm]{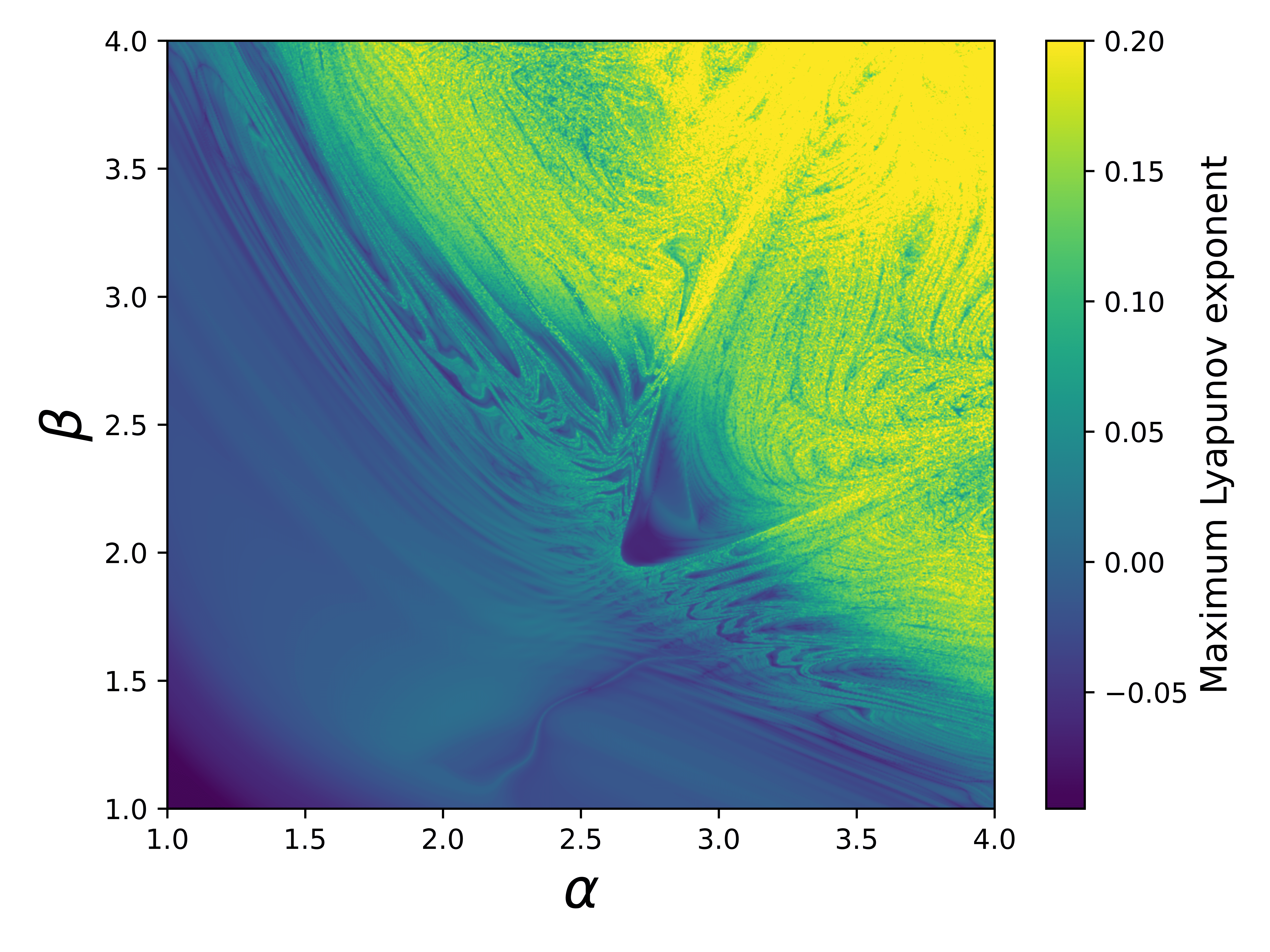}
\caption{2D bifurcation diagram $(\alpha, \beta)$ for $\alpha, \beta \in (1, 4)$, $\mu = \nu = 0.1$, $\sigma = \rho = 0.5$. The simulation returns the MLE after $100$ iterations on a $750 \times 750$ grid for the parameters.}
\label{2Dbif_alfa_beta_m01s05}
\end{figure}

The attractor found in Figure \ref{attractor_tot} arises in a regime of identical perturbations and external currents, so it is interesting to analyze the role of the symmetry between the two perturbations $\mu$, $\nu$. The bifurcation diagram shown in Figure \ref{2Dbif_mu_nu_a11b4} provides this kind of visualization, for $\mu, \nu \in (0,1)$, $\alpha = 1.1$, $\beta = 4$, $\sigma = \rho = 0.5$.

\begin{figure}[h!]
\centering
\includegraphics[width=8.5cm]{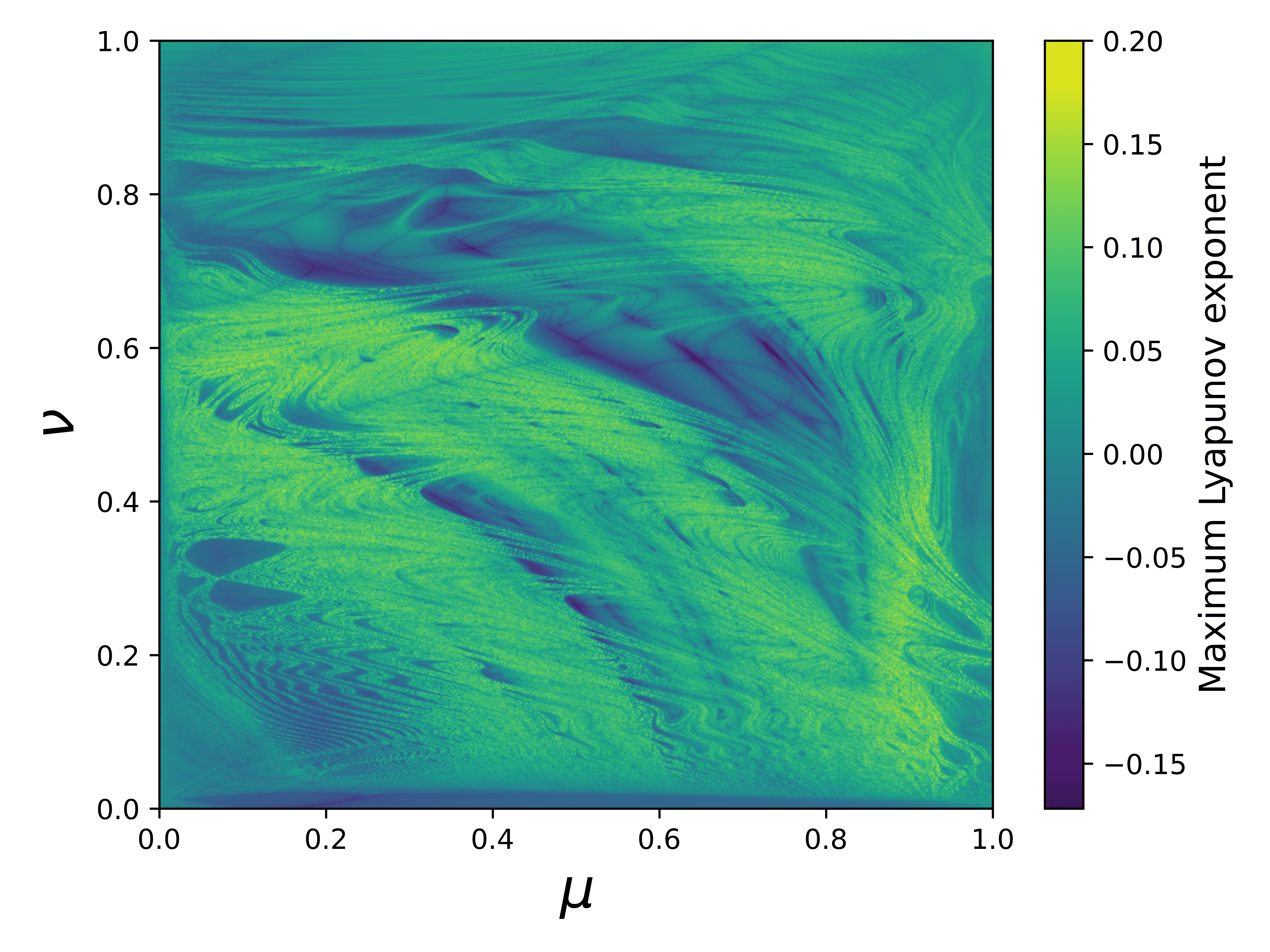}
\caption{2D bifurcation diagram $(\mu, \nu)$ for $\mu, \nu \in (0,1)$, $\alpha = 1.1$, $\beta = 4$, $\sigma = \rho = 0.5$. The simulation returns the MLE after $100$ iterations on a $750 \times 750$ grid for the parameters.}
\label{2Dbif_mu_nu_a11b4}
\end{figure}

In Figure \ref{2Dbif_mu_alfa_b4n01} it is presented the 2D bifurcation diagram $(\mu, \alpha)$ for $\mu \in (0, 1)$, $\alpha \in (1, 4)$, $\beta = 4$, $\nu = 0.1$, $\sigma = \rho = 0.5$, that allows an analysis of the role of a single perturbation combined with the value of the membrane potential $\alpha$. Analogously, in Figure \ref{2Dbif_nu_alfa_b4m01} it is presented the 2D bifurcation diagram $(\nu, \alpha)$ for $\nu \in (0, 1)$, $\alpha \in (1, 4)$, $\beta = 4$, $\mu = 0.1$, $\sigma = \rho = 0.5$.

\begin{figure}[h!]
\centering
\subfloat[][\label{2Dbif_mu_alfa_b4n01}]
{\includegraphics[width=.45\textwidth]{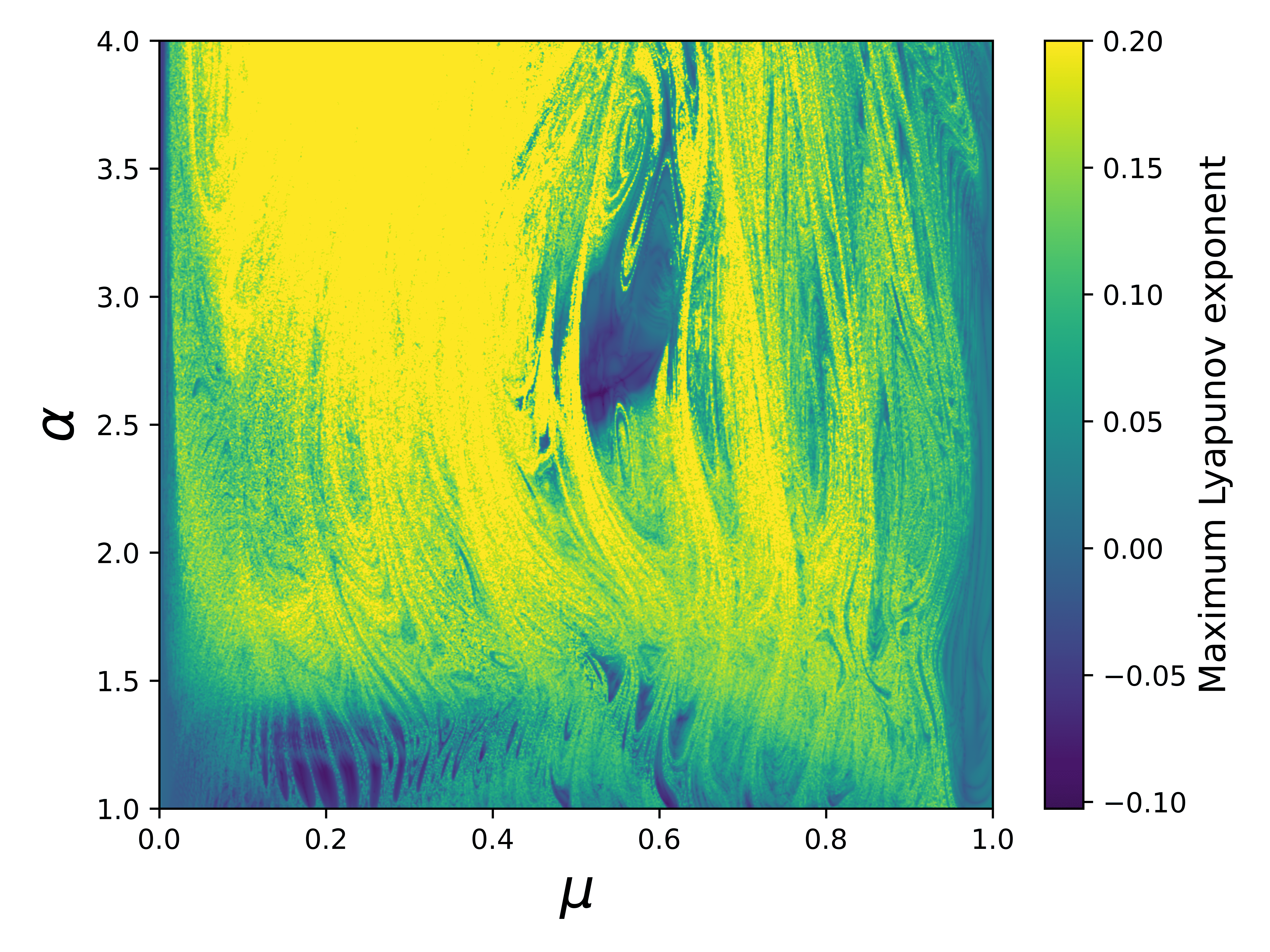}} \quad
\subfloat[][\label{2Dbif_nu_alfa_b4m01}]
{\includegraphics[width=.45\textwidth]{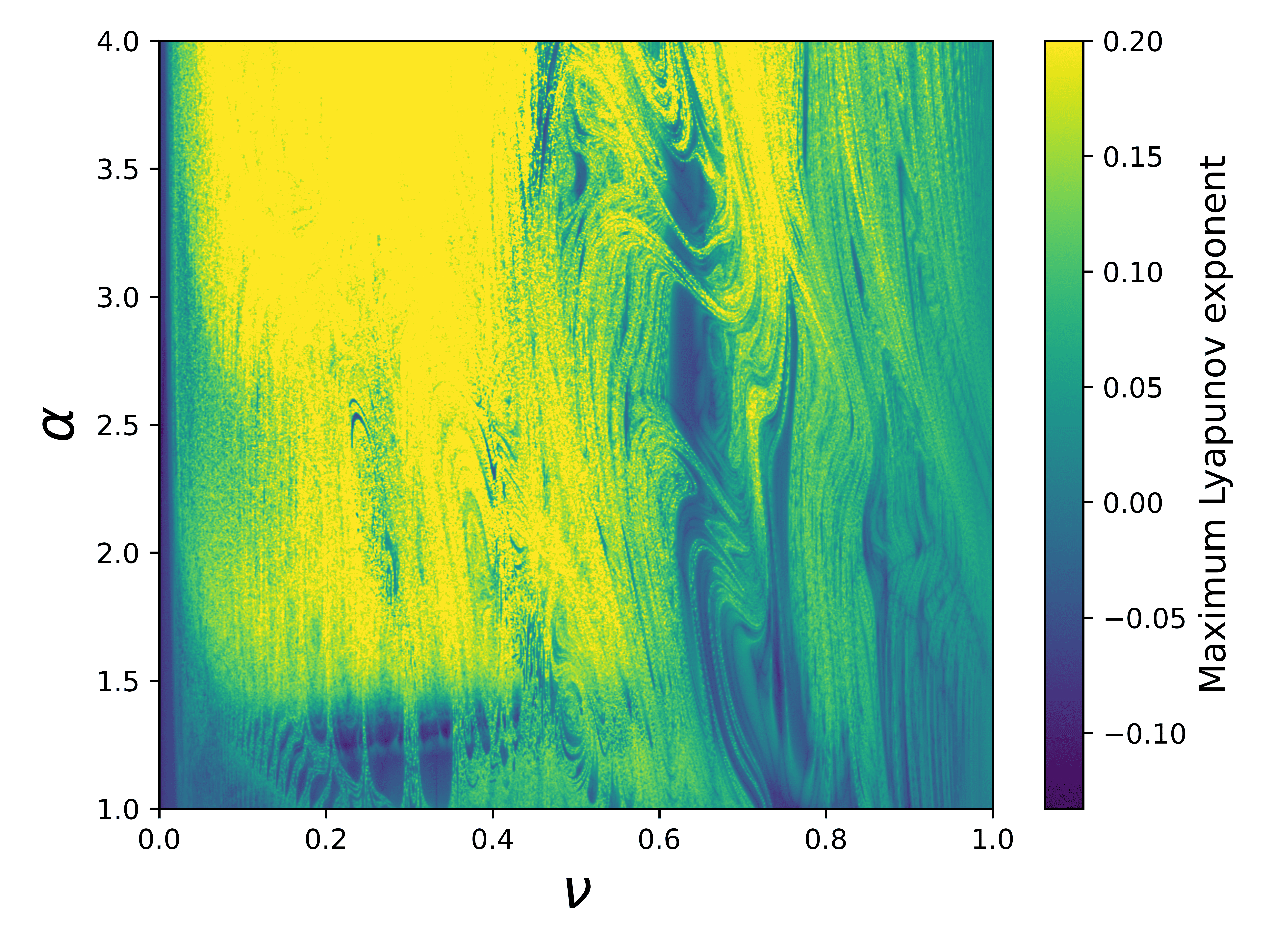}} \\
\caption{The simulations return the MLE after $100$ iterations on a $750 \times 750$ grid for the parameters. a) 2D bifurcation diagram $(\mu, \alpha)$ for $\mu \in (0, 1)$, $\alpha \in (1, 4)$, $\beta = 4$, $\nu = 0.1$, $\sigma = \rho = 0.5$. b) 2D bifurcation diagram $(\nu, \alpha)$ for $\nu \in (0, 1)$, $\alpha \in (1, 4)$, $\beta = 4$, $\mu = 0.1$, $\sigma = \rho = 0.5$.}
\label{2Dbif_mu_alfa_nu_alfa}
\end{figure}

Finally, in Figure \ref{2Dbif_sigma_rho_a11b4m01} it is presented the 2D bifurcation diagram $(\sigma, \rho)$ for $\sigma, \rho \in (0, 1)$, $\alpha = 1.1$, $\beta = 4$, $\mu = \nu = 0.1$, that allows an analysis on the role played by the two external currents $\sigma$, $\rho$.

\begin{figure}[h!]
\centering
\includegraphics[width=8.5cm]{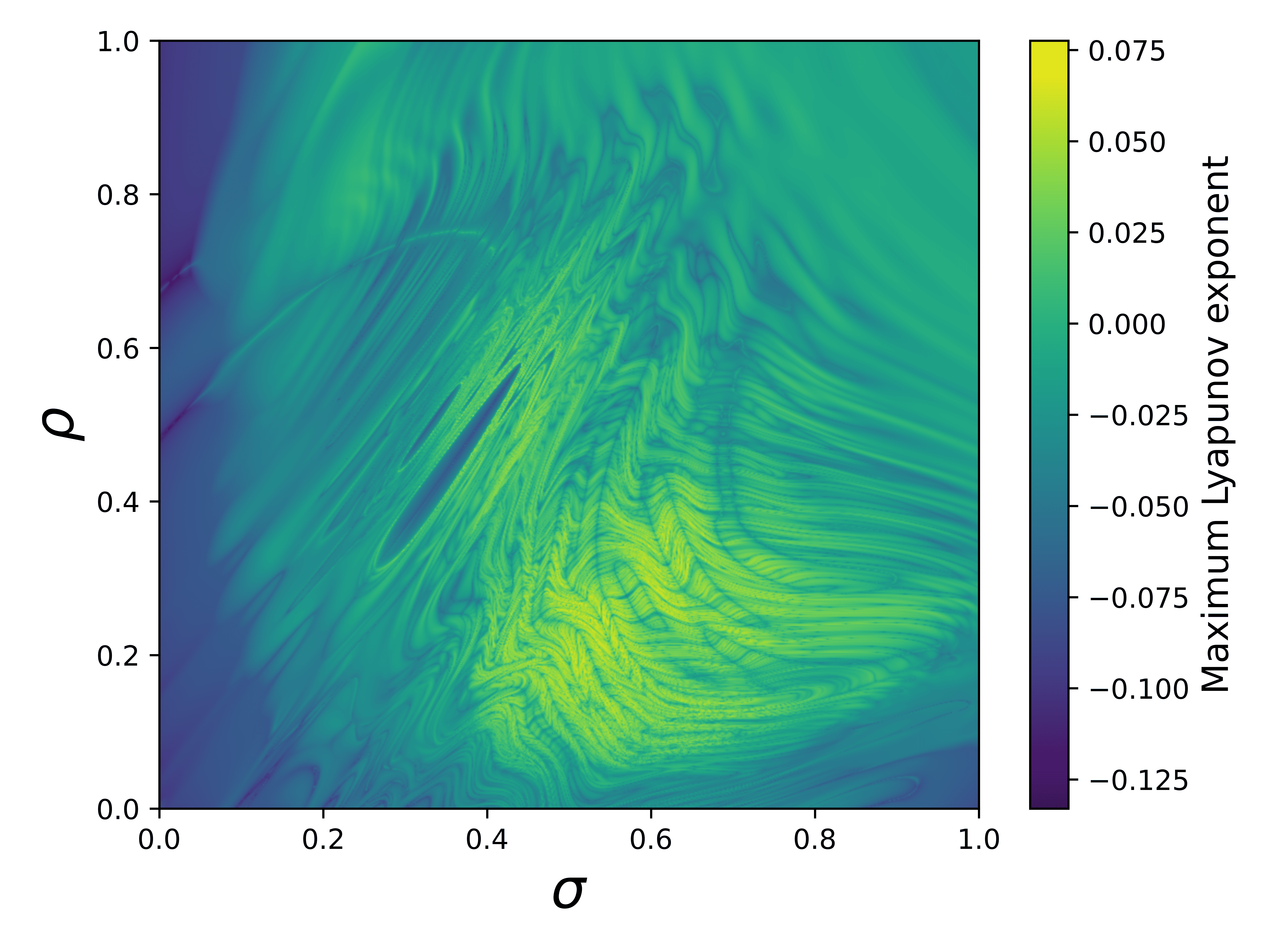}
\caption{2D bifurcation diagram $(\sigma, \rho)$ for $\sigma, \rho \in (0, 1)$, $\alpha = 1.1$, $\beta = 4$, $\mu = \nu = 0.1$. The simulation returns the MLE after $100$ iterations on a $750 \times 750$ grid for the parameters.}
\label{2Dbif_sigma_rho_a11b4m01}
\end{figure}

We can notice that the bifurcation diagram $(\sigma, \rho)$ shown in Figure \ref{2Dbif_sigma_rho_a11b4m01} shows a very small (positive) MLE on the entire grid $(\sigma, \rho) \in (0, 1) \times (0, 1)$, and this is completely reasonable in the context of the other simulations. Indeed, if we overlay all the diagrams around the parameter set associated to the attractor shown in Figure \ref{attractor_tot}, we see that the associated MLE is of the order $\sim 0.05$. This fact tells us that the strange attractor found in Figure \ref{attractor_tot} arises in a relatively weak chaotic regime. This is confirmed by the Lyapunov exponents spectra that, for this choice on the parameters, show just a positive Lyapunov exponent together with vanishing or negative values for the remaining three, i.e. we have expansion along one direction and contraction along the other three directions in the four-dimensional phase space.

Another consideration can be taken into account. The mean magnitude of the MLE in Figure \ref{2Dbif_sigma_rho_a11b4m01} suggests that, as for the two-dimensional Rulkov map \eqref{mapR}, the external currents do not play a dominant role in amplifying or suppressing the chaotic behavior for a given choice on membrane potentials and perturbations on the slow variables. Analytically, this follows immediately from the definition of the Jacobian, that depends on $\sigma$, $\rho$ only implicitly through the orbit. We may argue that this fact provides further support for the model: even though \eqref{mapC} is not a classical additive coupling that could be made arbitrary small, it seems to preserve the essential phenomenology of the original Rulkov model for one neuron.

Finally, we present basins of attraction of \eqref{Rulkov4D} for a choice on parameters that numerically suggest the arising of multistable states \cite{Pisarchik2014}, that is the coexistence of two or more different attractors solely depending on the initial conditions. This kind of behavior is confirmed by the fractal structure exhibited by the basins of attraction presented in Figure \ref{basins_alfa1205_res1000}. For this simulation, we implement the Julia script provided in \cite{Wagemakers2025}, taking $\alpha = 1.205$, $\beta = 4$, $\mu = \nu = 0.1$, $\sigma = \rho = 0.5$, $y_0 = \omega_0 = 0$ and $1.000$ initial values for $x_0, z_0 \in [-1, 1]$. We remark that the basins of attraction shown in Figure \ref{basins_alfa1205_res1000} correspond to three different regimes of transient chaos, generally followed by periodic orbits reached after a very large amount of iterations.

\begin{figure}[h!]
\centering
\includegraphics[width=8.5cm]{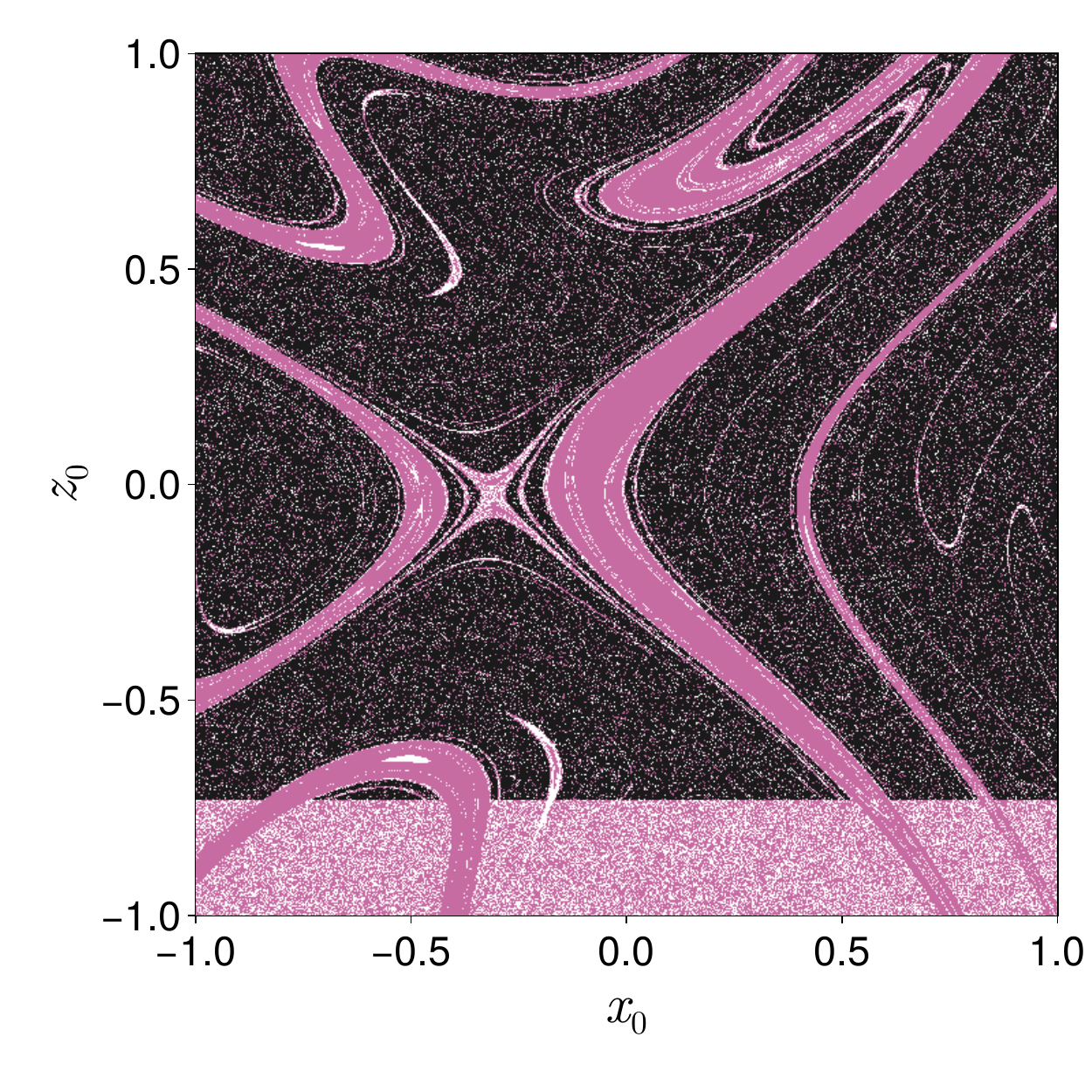}
\caption{Basins of attraction $(x_0, z_0)$ for $\alpha = 1.205$, $\beta = 4$, $\mu = \nu = 0.1$, $\sigma = \rho = 0.5$, $y_0 = \omega_0 = 0$. The simulation returns the three basins of attraction associated to $1.000$ initial conditions $x_0, z_0 \in [-1, 1]$, associated to three regimes of transient chaotic behavior.}
\label{basins_alfa1205_res1000}
\end{figure}

\section{Conclusions}\label{sec_conclusions}
In this paper we introduced a cross-coupling of two Rulkov neural maps, that is structurally different from the standard diffusive couplings that are analyzed in the literature. We proposed a heuristic biological interpretation concerning the transition from small to non-small values on the perturbations acting on the slow-variables, conjecturing a neural description of parameters regimes that are usually discarded. We analytically proved that the cross-coupling of two identical neurons admits a generalized snap-back repeller in the sense of Shiraiwa-Kurata, provided that a snap-back repeller in the sense of Marotto exists for the original two-dimensional map. We presented numerical simulations for the cross-coupling of two chaotic Rulkov maps, confirming the general chaotic behavior of the four-dimensional system for general choices on the parameters through analysis of time series, spectra of Lyapunov exponents, 1D and 2D bifurcation diagrams together with fractal basins of attraction. Furthermore, we numerically suggested the arising of a strange attractor by computing a non-integer Kaplan-Yorke dimension, that has been compared to some limit cases of periodic and ergodic orbits.

Even though the analytical formulation proposed in this work does not conform to orthodox neural models, the theoretical and computational insights provided by these studies suggest that the cross-coupling may be considered as a good dynamical model for the coupling of two neurons, also exhibiting some peculiar and interesting properties.

We outline some possible developments of this work.

Analytically, an extension of Theorem \ref{theo_cross} to the cross-coupling of different Rulkov neurons would be desirable.

Extensive analytical and numerical studies are needed to characterize the synchronization regimes \cite{Boccaletti2002} of the system \eqref{mapC}, in similar fashion to other couplings that have been analyzed in literature. The system \eqref{Rulkov4D} admits a complete synchronization of the subsystems $(x, y)$, $(z, \omega)$ when $\alpha = \beta$, $\mu = \nu$, $\sigma = \rho$, for which the subsystems $(x, y)$, $(z, \omega)$ are synchronized if they are initially synchronized, i.e. $(x_0, y_0) = (z_0, \omega_0)$. However, this is not the only synchronization regime (see e.g. Figure \ref{kaplan_example1}) and deriving precise regimes for complete and generalized synchronization does not represents an easy task, requiring a deep and aware use of the cocycle theory \cite{Arnold1998} that goes beyond the aim of this work.

We may think to explicitly construct an absorbing set for \eqref{Rulkov4D} by applying to \eqref{mapR} analytical approaches analogous to those presented in \cite{Pilarczyk2024} for the Chialvo map, then proving through Theorem \ref{theo_Ortega} the existence of an absorbing set for the cross-coupling. Furthermore, we have numerical evidences of a crisis-induced intermittency exhibited by the strange attractor of the system, that alternatively appears and disappears under very slight changes on the parameters of the system. The genesis of this behavior may be addressed using similar techniques to those presented in \cite{Tanaka2005}.

Another interesting development that we would outline concerns an extension of the presented results to a cross-coupling of a large number of Rulkov neurons. As a natural example, we propose the following one:
\beq\begin{cases}
x_{n+1}^{(j)} = \sum_{i = 1, i \ne j}^N \alpha_i f( x_n^{(i)} ) + y_n^{(j)} \\
y_{n+1}^{(j)} = y_n^{(j)} - \mu_j ( x_n^{(j)} - \sigma_j ) \,, \quad j = 1 \,, \dots \,, N \gg 1 \,.
\end{cases}\eeq
Even if Theorem \ref{theo_cross} may be extended to this setting, numerical studies of this high-dimensional system require more powerful computational tools and a more delicate numerical analysis. An important direction for future work will be the comparison of the model predictions with real EEG recordings, providing a first step toward the validation of the proposed framework on experimental data.
\\ \\ \textbf{Declaration of competing interests} \\The author has no conflicts to disclose.
\\ \\ \textbf{Funding} \\This research was funded by the University of Ferrara, FAR 2024.
\\ \\ \textbf{Data availability} \\The data that supports the funding of this study are available within the article.
\\ \\ \textbf{Declaration of generative AI and AI-assisted technologies in the manuscript preparation process} \\During the preparation of this work, the author used Google Gemini (Large Language Model) in order to assist with English language editing, basic Python scripting for numerical visualization and preliminary verification of symbolic calculations. After using this tool, the author reviewed and edited the content as needed and takes full responsibility for the content of the publication.
\\ \\ \textbf{Acknowledgments} \\The author is deeply grateful to prof. Rafael Ortega Ríos (University of Granada) for useful discussions and suggestions concerning the analytical contents of the article. The author is also deeply grateful to prof. Vincenzo Coscia (University of Ferrara) for his constant support and encouragement.
\\ \\ \textbf{ORCID} \\Stefano Disca - \url{https://orcid.org/0009-0000-6511-5608}

\appendix
\section{Computational details in the proof of Theorem \ref{theo_cross}}\label{sec_app}
In this appendix we present the computational details omitted in the proof of Theorem \ref{theo_cross}, concerning the transversal intersection between the stable and unstable manifolds of the system \eqref{mapC}. The symbolic computations were performed using MATLAB R2026a with the Symbolic Math Toolbox.

For $K < -1 - \frac{\mu}{2}$, we obtain the following expression for the matrix defined in \eqref{eq_M}:
\beq\label{eq_M_s1u3}
\begin{pmatrix}
\gamma_9 - \frac{1}{\mu} & 1 - \gamma_{15} - \gamma_{13} \gamma_{11} & \gamma_{13} \gamma_{10} - \gamma_{15} - 1 & 1 - \gamma_{15} - \gamma_{13} \gamma_8 \\
-1 & \frac{2 \alpha \mu q_{0z} \gamma_{11} }{\gamma_2 } - \mu \gamma_{11} - \mu + 1 & \mu + \mu \gamma_{10} + \frac{2 \alpha \mu q_{0z} \gamma_{10} }{\gamma_2 } - 1 & \frac{2 \alpha \mu q_{0z} \gamma_8 }{\gamma_2 } - \mu \gamma_8 - \mu + 1 \\
\frac{1}{\mu} - \gamma_9 & 1 - \gamma_{14} - \gamma_{12} \gamma_{11} & \gamma_{14} - \gamma_{12} \gamma_{10} + 1 & 1 - \gamma_{14} - \gamma_{12} \gamma_8 \\
1 & \frac{2 \alpha \mu q_{0x} \gamma_{11} }{\gamma_3 } - \mu \gamma_{11} - \mu + 1 & 1 - \mu \gamma_{10} - \frac{2 \alpha \mu q_{0x} \gamma_{10} }{\gamma_3 } - \mu & \frac{2 \alpha \mu q_{0x} \gamma_8 }{\gamma_3 } - \mu \gamma_8 - \mu + 1
\end{pmatrix} \,.
\eeq

For $-1 - \frac{\mu}{2} < K < \mu -1$, we obtain the following expression for the matrix defined in \eqref{eq_M}:
\beq\label{eq_M_s2u2}
\begin{pmatrix}
\delta_{13} & -\delta_{12} & 1 - \frac{\delta_{15} \delta_{11} }{ \delta_1 } - \delta_{17} & \frac{\delta_{15} \delta_{10} }{\delta_1} - \delta_{17} - 1 \\
1 & -1 & \frac{2 \alpha \mu q_{0z} \delta_{11} }{\delta_2 \delta_1 } - \frac{\mu \delta_{11} }{\delta_1} - \mu + 1 & \mu + \frac{\mu \delta_{10} }{\delta_1} + \frac{2 \alpha \mu q_{0z} \delta_{10} }{\delta_2 \delta_1 } - 1 \\
\delta_{13} & \delta_{12} & 1 - \frac{\delta_{14} \delta_{11} }{\delta_1} - \delta_{16} & \delta_{16} - \frac{\delta_{14} \delta_{10} }{\delta_1} + 1 \\
1 & 1 & \frac{2 \alpha \mu q_{0x} \delta_{11} }{\delta_3 \delta_1 } - \frac{\mu \delta_{11} }{\delta_1} - \mu + 1 & 1 - \frac{\mu \delta_{10} }{\delta_1} - \frac{2 \alpha \mu q_{0x} \delta_{10} }{\delta_3 \delta_1} - \mu
\end{pmatrix} \,.
\eeq

The coefficients $\gamma_1, \dots, \gamma_{15}, \delta_1, \dots, \delta_{17}$ are defined in Table \ref{tab_gamma_delta}.

\begin{table}[h!]
\centering
\begin{tabular}{c c c c}
\toprule
Coefficient & Value & Coefficient & Value \\
\midrule
$\gamma_1$ & $2 \mu (\sigma^2 + 1)^2$ & $\delta_1$ & $2 \mu ( \sigma^2 + 1 )^2$ \\
$\gamma_2$ & $( q_{0z}^2 + 1 )^2$ & $\delta_2$ & $( q_{0z}^2 + 1 )^2$ \\
$\gamma_3$ & $( q_{0x}^2 + 1 )^2$ & $\delta_3$ & $( q_{0x}^2 + 1 )^2 $ \\
$\gamma_4$ & $q_{0y} +\frac{\alpha}{ q_{0z}^2 +1}$ & $\delta_4$ & $q_{0y} q_{0z}^2 + \alpha + q_{0y}$ \\
$\gamma_5$ & $q_{0w} +\frac{\alpha}{q_{0x}^2 +1}$ & $\delta_5$ & $q_{0w} q_{0x}^2 + \alpha + q_{0w}$ \\
$\gamma_6$ & $(\sigma^2 + 1)^2 \sqrt{ (1+K)^2 - 4\mu }$ & $\delta_6$ & $(\sigma^2 + 1)^2 \sqrt{ (1+K)^2 - 4\mu }$ \\
$\gamma_7$ & $(\sigma^2 + 1)^2 \sqrt{ (1-K)^2 - 4\mu }$ & $\delta_7$ & $(\sigma^2 + 1)^2 \sqrt{ (1-K)^2 - 4\mu }$ \\
$\gamma_8$ & $\frac{1}{2\mu} + \frac{\gamma_6 + 2 \alpha \sigma}{\gamma_1}$ & $\delta_8$ & $( \frac{ \delta_4^2 }{\delta_2 }+ 1 )^2$ \\
$\gamma_9$ & $\frac{1}{2\mu} + \frac{ 2 \alpha \sigma - \gamma_7}{\gamma_1}$ & $\delta_9$ & $( \frac{ \delta_5^2 }{\delta_3 }+ 1 )^2$ \\
$\gamma_{10}$ & $\gamma_8 - \frac{ 4 \alpha \sigma + \gamma_6 + \gamma_7 }{\gamma_1}$ & $\delta_{10}$ & $\frac{\delta_1}{2 \mu} - \delta_7 - 2 \alpha \sigma$ \\
$\gamma_{11}$ & $\gamma_8 - \frac{2 \gamma_6}{\gamma_1}$ & $\delta_{11}$ & $\frac{\delta_1}{2\mu} - \delta_6 + 2 \alpha \sigma$ \\
$\gamma_{12}$ & $\mu - \frac{4 \alpha^2 q_{0z} \gamma_4 }{\gamma_2 ( \gamma_4^2 + 1 )^2 }$ & $\delta_{12}$ & $\frac{1}{2 \mu} + \frac{ \delta_7 - 2 \alpha \sigma }{\delta_1}$ \\
$\gamma_{13}$ & $\mu - \frac{4 \alpha^2 q_{0x} \gamma_5 }{\gamma_3 ( \gamma_5^2 + 1 )^2 }$ & $\delta_{13}$ & $\frac{1}{2 \mu} + \frac{ \delta_6 + 2 \alpha \sigma }{\delta_1}$ \\
$\gamma_{14}$ & $\frac{2 \alpha \gamma_4 }{ ( \gamma_4^2 + 1 )^2 }$ & $\delta_{14}$ & $\mu - \frac{4 \alpha^2 q_{0z} \delta_4 }{ ( q_{0z}^2 + 1 )^3 \delta_8 }$ \\
$\gamma_{15}$ & $\frac{2 \alpha \gamma_5 }{ ( \gamma_5^2 + 1 )^2 }$ & $\delta_{15}$ & $\mu - \frac{4 \alpha^2 q_{0x} \delta_5 }{ ( q_{0x}^2 + 1 )^3 \delta_9 }$ \\
 & & $\delta_{16}$ & $\frac{2 \alpha \delta_4 }{ ( q_{0z}^2 + 1 ) \delta_8 }$ \\
 & & $\delta_{17}$ & $\frac{2 \alpha \delta_5 }{ ( q_{0x}^2 + 1 ) \delta_9 }$ \\
\bottomrule
\end{tabular}
\caption{Coefficients $\gamma_i$, $i = 1, \dots, 15$ in \eqref{eq_M_s1u3} and $\delta_i$, $i = 1, \dots, 17$ in \eqref{eq_M_s2u2}.}
\label{tab_gamma_delta}
\end{table}

The computation of the rank through the native function $rank()$ returns $4$, but we remark that it does not rigorously prove that the rank is different from zero for any choice on the parameters. On the other hand, the explicit computation of the determinant is extremely prohibitive. For this reason, as specified in the thesis of Theorem \ref{theo_cross}, we are only able to state that the rank of the matrix defined in \eqref{eq_M} is different from zero up to a set of singular parameters $(\alpha, \mu, \sigma)$.

For the sake of completeness, we write the general functional equations that define the unstable manifold of the system \eqref{mapC}.

For $- 1 - \frac{\mu}{2} < K < \mu - 1$, $\dim W^u = \dim W^s = 2$. In the invariant manifold $W^u$, the stable coordinates $z, \omega$ can be expressed in terms of the unstable variables $x, y$, and $C^{(k)}(W^u) \subset W^u$. Therefore, we write $z_{k+1} = h_1(x_{k+1}, y_{k+1})$, $\omega_{k+1} = h_2(x_{k+1}, y_{k+1})$ and, substituting the definition of \eqref{mapC}, we obtain
\beq\label{eq_h1h2}
\begin{split}
&\begin{cases}
h_1 \bigg( \frac{\alpha}{1 + h_1^2(x, y)} + y \,, y - \mu(x - \sigma) \bigg) = h_2(x, y) + \frac{\alpha}{1 + x^2} \\
h_2 \bigg( \frac{\alpha}{1 + h_1^2(x, y)} + y \,, y - \mu(x - \sigma) \bigg) = h_2(x, y) -\mu h_1(x, y) + \mu \sigma \,,
\end{cases} \\
&h_1(p_0) = p_{0z} \,, \quad \nabla h_1(p_0) = (-K, 0) \,, \\
&h_2(p_0) = p_{0\omega} \,, \quad \nabla h_2(p_0) = (0, 0) \,,
\end{split}
\eeq
where the boundary conditions are imposed accordingly to the definition of \eqref{mapC}.

Analogously, for $K < - 1 - \frac{\mu}{2} < K$, $\dim W^u = 3$ and $\dim W^s = 1$. The invariant manifold $W^u$ is defined by the equation $\omega_{k+1} = h(x_{k+1}, y_{k+1}, \omega_{k+1})$, leading to
\beq\label{eq_h}
\begin{split}
&h(x, y, z) - \mu(z - \sigma) = h \bigg( \frac{\alpha}{1 + z^2} \,, y - \mu(x - \sigma) \,, \frac{\alpha}{1 + x^2} + h(x, y, z) \bigg) \,, \\
&h(p_0) = p_{0\omega} \,, \quad \nabla h(p_0) = (0, 0, -\mu) \,.
\end{split}
\eeq

The intrinsic cross-dependence between all the dynamical variables in the definition of \eqref{mapC} prevents any simplification in \eqref{eq_h1h2} and \eqref{eq_h}. Even within a quadratic approximation, $h_1, h_2$ and $h$ are supposed to be defined as complete quadratic forms, rendering the computational treatment of the resulting functional problems extremely prohibitive. This is the main motivation for the hypothesis introduced in Theorem \ref{theo_cross}.

\bibliographystyle{elsarticle-num}
\bibliography{Rulkov_biblio}

\end{document}